\documentclass[10pt,authoryear]{article}
\usepackage{appendix}
\usepackage{setspace}
\usepackage[font=small,  margin=2cm]{caption}
\usepackage[tbtags]{amsmath}
\usepackage{amsthm,amssymb,amsfonts}
\usepackage{natbib}
\usepackage{multirow}
\usepackage{lscape}
\usepackage{subfigure}
\usepackage{makecell}
\usepackage{exscale}
\usepackage{booktabs}
\usepackage{array}
\usepackage{fullpage}
\usepackage{url}
\usepackage{algorithm}
\usepackage{algpseudocode}
\usepackage{bm}
\usepackage{smile}
\usepackage{mathtools}
\usepackage{wrapfig}
\usepackage{lipsum}
\usepackage{threeparttable}
\usepackage{mathrsfs}
\usepackage{relsize}
\usepackage{dsfont}
\usepackage{multirow}
\usepackage{apalike}
\usepackage[usenames,dvipsnames,svgnames,table]{xcolor}
\usepackage[colorlinks=true,linkcolor=blue,urlcolor=blue,citecolor=blue]{hyperref}

\newtheorem{assumption}{{\sc Assumption}}

\newcommand{\bargmin}{\mathop{\mathrm{arg\ min}}}
\newcommand{\bargmax}{\mathop{\mathrm{arg\ max}}}

\begin{document}

\title{\LARGE Quantile Factor Analysis for Large-dimensional Time Series with Statistical Guarantee
}

	\author{Yong He\thanks{ Institute for Financial Studies, Shandong University, Jinan, China; },~~Xin-Bing Kong\thanks{Nanjing Audit University, Nanjing, 211815, China; Email:{\tt xinbingkong@126.com}.},~~Long Yu\thanks{National University of Singapore, 117546, Singapore;},~~Peng Zhao\thanks{Jiangsu Normal University,
Xuzhou, Jiangsu, {\rm221116}, China}}	
	\date{}	
	\maketitle
Quantile is an important measure in finance and quality assessment in service industry. In this paper, we model the temporal and cross-sectional interactive effect of the quantiles of  large-dimensional time series by a latent quantile factor model. The factor loadings and scores are learnt with statistical guarantee via an iterative check-loss-minimization procedure.
Without any moment constraint on the idiosyncratic errors, we correctly identify the common and idiosyncratic components for each variable. We obtained the statistical convergence rates of the minimization estimators. Bahardur representations for the estimated factor loadings and scores are provided under some mild conditions. Moreover, a robust method is proposed to select the number of factors consistently. Simulation experiments checked the validity of the theory. Our analysis on a financial data set shows the superiority of learning quantile factors in portfolio allocation over other state-of-the-art methods that learn mean factors.

\vspace{2em}


\section{Introduction}

Factor models are widely used in practice such as biology, image processing, recommending system, economics and finance. The mathematical expression of a large dimensional static approximate factor model is
\begin{equation}\label{model}
(\bY_t)_{p\times 1}=\bL_{p\times r}(\bbf_t)_{r\times 1}+(\bepsilon_t)_{p\times 1}, \ t=1,\ldots, T,
\end{equation}
where $\bY_t$ is a $p$-dimensional vector observed at time $t$, $\bL$ is the factor loading matrix, $\bbf_t$ is a vector of factors at time $t$, and $\bepsilon_t$ is the idiosyncratic component that can be cross-sectionally weakly dependent. Recent years have seen increasing interest in statistical inference on model (\ref{model}). The approximate factor structure instead of the strict factor structure was introduced and studied in \cite{Chamberlain1983Arbitrage}. \cite{Bai2002Determining} presented information criterions to determining the number of factors under the framework of the static approximate factor model. \cite{Bai2003Inferential} further gave the asymptotic theory on the estimated factor loadings and scores. \cite{fan2013large} proposed a POET estimator of the large covariance matrix with factor structure.  \cite{stock2002forecast,Stock2002Macroeconomic} incorporated the factors into the autoregressive model to predict macroeconomic variables. \cite{onatski2009testing} provided a hypothesis testing procedure to a prefixed number of factors. \cite{Ahn2013Eigenvalue} proposed the eigenvalue-space-ratio estimators of the number of factors. {\cite{Trapani2018A} sequentially tested the divergence of eigenvalues and found a consistent estimate of the number of factors. } \cite{Kong2019Factor} established the theory of empirical processes of the series of  estimated common components and idiosyncratic components. With high-frequency data, \cite{Ait2017Using}, \cite{Pelger2018Large}, \cite{Chen2020The}, \cite{Kim2016Sparse}, \cite{Kong2017On,Kong2018On} and \cite{Kong2018Testing} extensively studied the continuous-time version of model (\ref{model}). In the seminal paper by \cite{forni2000generalized}, the authors proposed a generalized dynamic factor model that can accommodate a factor space of infinite dimension and the factors are loaded via linear filters. Adapting to the dynamic feature, \cite{hallin2007determining} developed an information criterion to estimate the number of factors. In this paper, we only consider robust estimation of the factors, loadings, and common and idiosyncratic components, under the static model, and leave extensions to the generalized dynamic factor model to our future work.



A basic requisite of the aforementioned papers is the finiteness of forth moment of idiosyncratic errors to obtain the convergence rate of the estimated factor loadings and scores. Theoretically, a natural question is ``{\it how to do factor analysis if the fourth moment or even the second moment does not exist}". In practice, many financial and macroeconomic variables have heavy-tailed distributions, and thus the assumption in most recent PCA-based factor analysis papers is violated. In finance, a stylized empirical fact of asset returns is leptokurtosis, c.f., Chapter 1 of \cite{Tancov2004Financial} and \cite{kong2015testing}. This motivates us to find a way to do factor analysis under model (\ref{model}) without any moment constraint on the idiosyncratic errors and with theoretical guarantee and computational feasibility.

To the best of our knowledge, few papers considered robust factor analysis without moment constraint. \cite{He2020Large} provided robust consistent estimates of the factor space and common and idiosyncratic components using an eigen-analysis of the spatial Kendal's tau matrix. However, it assumed a joint elliptical distribution for the factor vector and the large cross-section of idiosyncratic components, which rules out typical families of heavy-tailed distributions, such as stable distributions. \cite{Calzolari2018Estimating} assumed the stable distribution for independent factors and idiosyncratic noises and did factor analysis with indirect inference, but no asymptotic theory was established. We are aimed at giving a completely nonparametric approach and providing reliable asymptotic results for separating common and idiosyncratic components from each variable and for estimating the factor loadings and scores, as $p, T\rightarrow \infty$ simultaneously.


Our methodology is inspired by the equivalence of PCA and double least square estimation when there aren't missing values. {That is, the PCA-based estimators of the factor loadings and scores are identical to}
\begin{equation}\label{L2}
(\overline{\bL}, \overline{\bF})=\bargmin_{\bL, \bF}\left\{\sum_{i=1}^p\sum_{t=1}^T(y_{it}-\bl_i^{\prime}\bbf_t)^2\right\},
\end{equation}
up to some orthogonal transformations, where $y_{it}$ is the $(i,t)$-th entry of $\bY_t$, $\bL=(\bl_1,\ldots,\bl_p)^{\prime}$ and $\bF=(\bbf_1,\ldots,\bbf_T)$.
As in robust regression, we simply replace the quadratic loss function by the weighted absolute loss function. That being said,
\begin{equation}\label{L1}
(\hat{\bL}, \hat{\bF})=\bargmin_{\bL, \bF}\left\{\sum_{i=1}^p\sum_{t=1}^T\rho_{\tau}(y_{it}-\bl_i^{\prime}\bbf_t)\right\}=:\bargmin_{\bL, \bF}\|\bY-\bL\bF\|_{WL_1},
\end{equation}
where $\bY=(\bY_1,\ldots,\bY_T)$, $\hat{\bL}=(\hat{\bl}_1,\ldots,\hat{\bl}_p)^{\prime}$, $\hat{\bF}=(\hat{\bbf}_1,\ldots,\hat{\bbf}_T)$, and $\rho_{\tau}(x)=\{\tau-I(x\leq 0)\}x$. The loss function $\rho_{\tau}(x)$ puts weight $\tau$ $(\tau\in(0, 1))$ to the positive axis and $1-\tau$ to the negative axis. It is also named as check function in the literature. When $\tau=1/2$, it is simply the absolute loss function. The optimization solutions to (\ref{L2}) and (\ref{L1}) have the advantage that they are not much affected by the missing values in $\bY_t$'s compared with the PCA solution. This is because the PCA solution relies on the input of a sample covariance matrix. To calculate the sample covariance matrix, one needs to delete the $t$-th column of $\bY$ if $y_{it}$ is missing for some $i$ or impute $y_{it}$ with some extra effort. For optimizing the loss functions in (\ref{L2}) or (\ref{L1}), only the single loss term containing $y_{it}$ needs to be deleted when $y_{it}$ is missing. This advantage is advocated in machine learning area, such as image processing, c.f., \cite{Ke2005robust} and \cite{Aan2002robust}. However, no statistical theory had ever been presented in machine learning field. The major difficulty in deriving the asymptotic theory of the estimated factors and loadings via optimizing the weighted $L_1$ loss in (\ref{L1}) lies in three aspects. First, the minimizers of (\ref{L1}) have no closed form expression compared with the PCA solution (or equivalently the $L_2$ minimizer of (\ref{L2})); Second, the weighted $L_1$ loss function in (\ref{L1}) is not a jointly convex function of $\bL$ and $\bF$, which is totally different from the least weighted absolute deviation setting in quantile regressions; Third, there are a large number of parameters to be optimized in (\ref{L1}) as $p, T\rightarrow \infty$ simultaneously, which makes it hard to construct a small ball containing the true parameters in the parameter space in contrary to the typical derivation of the consistency of the robust regression estimators, c.f., \cite{Pollard1991Asymptotics} and \cite{Knight1998Asymptotics}.


The weighted $L_1$ minimization in (\ref{L1}) amounts to saying that the $\tau$-th quantiles of a large cross-section of asset returns are driven by the common factor vector $\bbf_t$ and the corresponding exposures are measured by the loading matrix $\bL$. And dynamically the quantiles of the return series are modeled by $\{\bl_i^{\prime}\bbf_t\}$ given latent $\bbf_t$. This implies the identifiability condition for the idiosyncratic components, $Q_{\tau}(\epsilon_{it}|\{\bbf_t^0\})=0$ for all $i=1,\ldots, p$ and $t=1,\ldots, T$, where $Q_{\tau}(X|Y)$ refers to the $\tau$-th quantile of $X$ given $Y$ and $\bbf_t^0$ is the true factor vector realized at time $t$. \cite{Ando2020Quantile} studied the quantile co-movement in financial market, but their theory requires at least finiteness of the first moment of the idiosyncratic errors.  In this paper, we derive the convergence rates of computationally feasible weighted $L_1$ estimators of the common components, factor loadings and scores. For the term computationally feasible estimator, it refers to an estimator in the algorithmic solution path after a number of alternating iterations, rather than the theoretically optimal minimizers. This is another originality of the present paper.
We show that up to some orthogonal transformations, the weighted $L_1$ estimators of the loadings converge at rate $({1}/{\sqrt{p}})\wedge ({\log{p}}/{\sqrt{T}})$, where the $\log{p}$ term stems from the aggregation of the estimation errors along the cross-sectional dimension in the solution path. Our results do not need any moment constraint on the idiosyncratic components. Under some mild conditions, we obtained the Barhadur representations of the estimated factor loadings and scores. The weighted $L_1$ estimation successfully separates the common and idiosyncratic components for each variable.


Related to the present paper, we are noticed most recently an interesting concurrent independent working paper by \cite{Chen2019Quantile}. They also proposed estimating the factor loadings and scores by (\ref{L1}) and got a similar conclusion that the finite forth moment constraints on the idiosyncratic errors can be relaxed, see the definition of $\overline{M}_{NT}(\theta)$ in Remark 1.1 of their paper, though we believe that their condition $E\rho_{\tau}(\epsilon_{it})<\infty$ can be removed completely. However, our paper differs from theirs in at least the following aspects. First, their paper considered the theoretical minimizer $(\hat{\bL}, \hat{\bF})$ of (\ref{L1}) while ours is concerned with directly the computationally feasible estimator (i.e., $(\tilde{\bL}, \tilde{\bF})$ in Algorithm 1 below), and for more details see the remarks after Algorithm 1; Second, their paper provided results on the summed squared errors of the estimated factor loadings, while ours on the maximum estimation error with the target of separating each variable to common and idiosyncratic components; third, they assumed the uniform boundedness of $\|\bbf_t^0\|$'s, while we only need some moment conditions on them; forth, their estimation of the number of factors is based on thresholding the eigenvalues while ours on maximizing eigenvalue ratios.


The present paper is arranged as follows. In Section \ref{assmain}, we present some setup assumptions and provide the main results of computationally feasible weighted $L_1$ estimators, realized by an alternating iterative algorithm to solve the non-convex objective function in (\ref{L1}). In addition, a robust method is proposed to estimate  the number of factors consistently. Extensive simulation studies and an empirical application are given in Section \ref{se} and Section \ref{em}, respectively. A brief conclusion and discussion on future works are given in Section \ref{condis}. All the technical proofs are relegated to the Appendix.

\section{Assumptions and main results}\label{assmain}

It is well known that the factor loadings and factors are only identifiable up to some orthogonal transformations. This gives the freedom to restrict the columns of the factor loading matrix to be orthogonal vectors spanning the same factor space. Notice also that the factor space spanned by the columns of $\bL$ is the same as that spanned by the $r$ principal components of $\bL Cov(\bbf_t) \bL^{\prime}$, without loss of generality and as in \cite{fan2013large}, {we assume that $\bL$ and $\bF$ have the canonical form in (\ref{canonical}) below.}

\begin{assumption}\label{ass1}
\begin{enumerate}
\item[(1)] The factor loading matrix $\bL$ and the factor series $\bbf_t$ satisfy
{\begin{equation}\label{canonical}
\bL^{\prime}\bL/p \ \mbox{is \ diagonal \ and } \ \bSigma_f=Cov(\bbf_t)=\Ib_r,
\end{equation} where $\Ib_r$ stands for the $r\times r$ identity matrix and the diagonal elements of $\bL^{\prime}\bL/p$ are bounded away from zero and infinity;}

\item[(2)] {$\{\bbf_t^0\}$ is a stationary and $\alpha$-mixing sequence of random vectors satisfying $E\|\bbf_t^0\|^4\leq C$ for some constant $C>0$, and $\sum^{\infty}_{n=1}\sqrt{\alpha(n)}<\infty$, where
\[
\alpha(n)=:\sup\Big\{\big|P(A\cap B)-P(A)P(B)\big|; A\in \sigma(\bbf_{-\infty},..., \bbf_{-n}), B\in\sigma(\bbf_k, k\geq 0)\Big\}.
\]}
\end{enumerate}
\end{assumption}

Given Assumption \ref{ass1}(1), the weighted $L_1$ minimization (\ref{L1}) can be done subject to
\begin{equation}\label{canonical1}
\bL^{\prime}\bL/p \ \text{is  diagonal  and } \frac{1}{T}\sum^T_{t=1}\bbf_t\bbf_t^{\prime}=\Ib_r.
\end{equation}
Assumption \ref{ass1}(1) assumed a strong factor condition saying that the signal strength of the common components grows at rate $\sqrt{p}$. This condition is mainly used to derive the second-order property of the estimators. For only the consistency, this might be relaxed to the weak factor condition that $\bL^{\prime}\bL/p^{\alpha}$ has bounded eigenvalues for some $0<\alpha<1$ as long as the common and idiosyncratic components are separable asymptotically. Assumption \ref{ass1}(2) is a standard assumption on the factor series, c.f., \cite{fan2013large} and the references therein.

\begin{assumption}\label{ass2}
\begin{equation}\label{median}
Q_{\tau}\left(\epsilon_{it}|\{\bbf_t^0\}\right)=0.
\end{equation}
\end{assumption}

Assumption \ref{ass2} is an identifiability condition for weighted $L_1$ optimization. When the factors are observable, it is simply the identifiability condition used in quantile regression. It is equivalent to stating $Q_{\tau}(y_{it}| \{\bbf_t^0\})=\bl_i^{\prime}\bbf_t^0$ which means the quantiles of a large cross-section of asset returns are driven by the common true factor vector $\bbf_t^0$ and the corresponding exposures are measured by the loading matrix $\bL$. This is not in accordance with  the classic CAPM theory which explains the mean cross-section excess returns via exposure to the value of the market portfolio. But the focus of the present paper is not on finance theory but a statistical investigation into the weighted $L_1$ estimators of the factor loadings, scores, and the common and idiosyncratic components under (\ref{L1}), (\ref{canonical}) and (\ref{median}).

Before presenting the next assumption on temporal and cross-sectional weak dependence on functionals of $\{\epsilon_{it}\}$, we introduce two sums of bounded functionals of $\{\epsilon_{it}\}$. Let
\[
H_1(\{\epsilon_{it}\}, p, T)=\sum^p_{i=1}\sum^T_{t=1}\Big\{(\sigma_{it}-\overline{\epsilon}_{it})\big[I(\overline{\epsilon}_{it}\leq \sigma_{it})-I(\overline{\epsilon}_{it}\leq 0)\big]-E_f(\sigma_{it}-\overline{\epsilon}_{it})\big[I(\overline{\epsilon}_{it}\leq \sigma_{it})-I(\overline{\epsilon}_{it}\leq 0)\big]\Big\},\\
\]
where $E_f$ stands for conditional expectation on $\{\bbf_t^0\}$ (the true factor vector), $\sigma_{it}$'s are bounded variables, $\overline{\epsilon}_{it}=\epsilon_{it}-\mu_{it}$ with $\mu_{it}$'s being fixed parameters.
Let $$
H_2(\{\epsilon_{it}\}, p, T)=\sum^p_{i=1}\sum^T_{t=1}c_{it}\Big\{I(\overline{\epsilon}_{it}\leq 0)-\tau-E_f\big[I(\overline{\epsilon}_{it}\leq 0)-\tau\big]\Big\},
$$
where $c_{it}$'s are bounded coefficients irrelevant to $\epsilon_{it}$'s.

\begin{assumption}\label{ass3}
\begin{enumerate}
\item[(1)] {$\epsilon_{it}$ has probability density function $h_{i}(x)$ satisfying $\min_ih_i(x)>0$ for all $x\in R$. The derivative function $\dot{h}_i(x)$ of $h_i(x)$ is bounded uniformly in $i$. For $M$ large enough, $\min_i\inf_{|x|\leq M}h_i(x)>c>0$ for some constant $c$ and $h_i(x)$ does not increase as $|x|\rightarrow\infty$ for $|x|>M$;}

\item[(2)] {$E_f\Big\{H_1(\{\epsilon_{it}\}, p, T)/\sqrt{pT\max_{i,t}|\sigma_{it}|^3}\Big\}^2\leq C$ and $E_f\Big\{H_2(\{\epsilon_{it}\}, p, T)/\sqrt{pT\max_{it}c_{it}^2}\Big\}^2\leq C$ for some constant $C$.}
\end{enumerate}
\end{assumption}

Assumption \ref{ass3}(1) is a regular condition on the distribution functions of the idiosyncratic components. It assumes that the probability density functions of $\epsilon_{it}$'s have uniform support. The assumption does not impose any moment constraint on $\epsilon_{it}$'s. The moment condition in Assumption \ref{ass3}(2) assumes that a series of bounded functions of $\epsilon_{it}$'s are weakly correlated temporally and cross-sectionally, under which the $\sqrt{pT\max_{i,t}|\sigma_{it}|^3}$ and $\sqrt{pT\max_{it}c_{it}^2}$ give the scales of $H_1(\{\epsilon_{it}\}, p, T)$ and $H_2(\{\epsilon_{it}\}, p, T)$, respectively. This is satisfied when $\epsilon_{it}$'s are independent given $\{\bbf_t^0\}$.

Different from the optimization problem (\ref{L2}), problem (\ref{L1}) has no explicit closed form solution. Yet the SVD algorithm designed for problem (\ref{L2}) with no missing values is not applicable to solving problem (\ref{L1}). To be computationally feasible, we introduce an alternating iterative algorithm to solve the optimization problem (\ref{L1}). Although the objective function in (\ref{L1}) is in general non-convex jointly in all parameters, it is indeed convex in $\bL$ (or $\bF$) when $\bF$ (or $\bL$) is fixed in advance. The above fact motivates to minimize the loss function alternatively over $\bL$ and $\bF$, each time optimizing one argument while keeping the other fixed. The alternating optimization steps can be solved by linear programming or gradient descent schemes. The detailed algorithm is presented in Algorithm \ref{alg:first}.

\begin{algorithm}[H]
	\caption{Iterative Algorithm for Robust Factor Analysis}\label{alg:first}
	{\bf Input:} $\cD=\{\by_t,t=1,\ldots, T\}$\\
	{\bf Output:} Alternating Iterative Estimates of the factor loadings and scores, i.e., $\tilde \bL$, $\tilde \bF$
	\begin{algorithmic}[1]
		\State Initialization: $k=0$; Set $\bL^{(0)}=(l_{ij}^{(0)})$ so that (\ref{canonical1}) is satisfied.

\\

$\hat{\bF}^{(k)}=\bargmin_{\bF}\|\bY-\hat{\bL}^{(k-1)}\bF\|_{WL_1}$, where $\hat{\bL}^{(0)}=\bL^{(0)}$, and then {transform $\hat{\bF}^{(k)}$ so that (\ref{canonical1}) is satisfied.} \\

$\hat{\bL}^{(k)}=\bargmin_{\bL}\|\bY-\bL\hat{\bF}^{(k)}\|_{WL_1}$ and then transform $\hat{\bL}^{(k)}$ so that (\ref{canonical1}) is satisfied. \\

		Repeat Steps 2-3 until convergence. \\
        Output $\tilde{\bL}=\hat{\bL}^{(K)}$ and $\tilde\bF=\hat{\bF}^{(K)}$ as the final estimates of the factor loading and score matrices when the convergence condition is met.
	\end{algorithmic}
\end{algorithm}

Algorithm \ref{alg:first} amounts to alternatively carrying out cross-sectional quantile regression on factors and serial quantile regression on loadings, starting from some initial guess of $\bL$. One could also start from an initial guess of $\bF$ and alternating the serial and cross-section quantile regression iteratively. To reduce the sensitivity in the initial parameter values, we can try a set of different initial parameters and choose the solution resulting in lowest loss. As for the convergence criterion, denote the factor loading and score matrices at the $k$-th step as $\hat{\bL}^{(k)}=(\hat{l}_{ij}^{(k)})=(\hat{\bl}_1^{(k)},\ldots,\hat{\bl}_p^{(k)})^{\prime}$, $\hat{\bF}^{(k)}=(\hat{f}_{ij}^{(k)})=(\hat{\bbf}_1^{(k)},\ldots,\hat{\bbf}_T^{(k)})$ and let $\bC^{(k)}=\hat{\bL}^{(k)}\hat{\bF}^{(k)}=(C_{ij}^{(k)})$. In our simulation studies, the iteration is terminated with a prefixed finite number of alternating steps or when
\begin{equation}\label{equ:convergecondition}
\sum_{i}\sum_{j}\big|C_{ij}^{(K)}-C_{ij}^{(K-1)}\big|/\big(pT|C_{ij}^{(K-1)}|\big)=o\Big(\frac{\log{p}}{\sqrt{T}}+\frac{1}{\sqrt{p}}\Big),
\end{equation}
which means that the average relative iteration error for computing the common components are small enough compared with the estimation error theoretically obtained in Theorem \ref{th1} below. Our simulation experience shows that the above accuracy tolerance condition is always met within a finite number of iterations and $\{\hat{\bl}_i^{(k)}\}_{k=1}^K$,  $\{\hat{\bbf}_t^{(k)}\}_{k=1}^K$ form a solution path of the algorithm. The alternating iterative estimators are simply the ending-step solutions of the path. Notice that $\tilde{\bL}$ and $\tilde{\bF}$ are generally different from $\hat{\bL}$ and $\hat{\bF}$. $\tilde{\bL}$ and $\tilde{\bF}$ are computationally feasible while $\hat{\bL}$ and $\hat{\bF}$ are only theoretical minimizers. Therefore $\hat{\bl}_i$'s and $\hat{\bbf}_t$'s incur two sources of errors, the computing error for a fixed sample measured by the discrepancy between $(\hat{\bl}_i, \hat{\bbf}_t)$ and $(\hat{\bl}^{(K)}_i, \hat{\bbf}^{(K)}_t)$, and the statistical estimation error due to the sampling randomness. Thus instead of investigating into the asymptotics of the theoretical minimizers having unknown computing error, we are concerned with the asymptotics of the feasible alternating iterative estimators. Our theory below shows that the $\tilde{\bl}_i$'s and $\tilde{\bbf}_t$'s correctly identifies the true loadings and realized factors up to orthogonal transformations, and that $C^{(K)}_{ij}$'s consistently match the true common components.

To successfully implement the alternating iterative algorithm, we need a slightly stronger version of Assumption \ref{ass1} to regularize the parameter space.

\noindent{\sc Assumption {\bf 1}'} \
{\it Assumption \ref{ass1} holds and} {\begin{enumerate}
\item[\text{(1)}] {\it the eigenvalues of $\bL^{\prime}\bH\bL/p$ and $\bL^{\prime}\bH\bL^0/p$ are bounded away from zero and infinity, where $\bH=\text{diag}\big\{h_i(\bl_i^{\prime}\bbf_t-\bl_i^{0\prime}\bbf_t^0)\big\}$ is a $p\times p$ diagonal matrix and $\bL^0=(\bl_1^0,\ldots,\bl_p^0)^{\prime}$; }
\item[\text{(2)}] {\it $\max_i\|\bl_i\|\leq C$ for some generic constant $C$.}
\end{enumerate}
}

Assumption {\bf 1}'(1) demonstrates that the loading parameters span a full rank-$r$ space after being normalized by the probability density of $\epsilon_{it}$'s, and the spaces spanned by $\bL$ and $\bL^0$ are not orthogonal after the same normalization. Assumption {\bf 1}'(2) restricts that the loadings for each variable are not explosive. We remark that this assumption is not minimal. As a first attempt to establish the asymptotic theory for the robust factor analysis and for technical simplicity, we assume this condition in the present paper. We leave extending the theory to more general setup to our future work.

\begin{assumption}\label{ass4}
\[
\begin{array}{rll}
E\exp\Big\{H_1(\{\epsilon_{it}\}, p, 1)/\sqrt{p\max_{i,t}|\sigma_{it}|^3}\Big\}\leq C, &\  E\exp\Big\{H_1(\{\epsilon_{it}\}, 1, T)/\sqrt{T\max_{i,t}|\sigma_{it}|^3}\Big\}\leq C, \\
E\exp\Big\{H_2(\{\epsilon_{it}\}, p, 1)/\sqrt{p\max_{i,t}c_{it}^2}\Big\} \leq C, &\  E\exp\Big\{H_2(\{\epsilon_{it}\}, 1, T)/\sqrt{T\max_{i,t}c_{it}^2}\Big\}\leq C.
\end{array}
\]
\end{assumption}
Assumption \ref{ass4} is satisfied if $\big\{\epsilon_{it}|\{\bbf_t^0\}\big\}$ are independent arrays due to the boundedness of the summands of $H_1(\{\epsilon_{it}\}, p, T)$ and $H_2(\{\epsilon_{it}\}, p, T)$. Next assumption provides the conditions on the increasing orders of $p$ and $T$.

\begin{assumption}\label{ass5}
$$
\frac{\log{p}}{\sqrt{T}}+\frac{\log{T}}{p^{1/4}\log{p}}=o(1).
$$
\end{assumption}
Assumption \ref{ass5} assumes that $p$ (or $T$) can not be exponentially large relative to $T$ (or $p$). The reason is that $\hat{\bbf}_t^{(k)}$'s (or $\hat{\bl}_i^{(k-1)}$'s) are required to converge uniformly in $t$ (or $i$) to guarantee the convergence of $\hat{\bl}_i^{(k)}$ (or $\hat{\bbf}_t^{(k)}$) in Algorithm \ref{alg:first}.

Now we state our theoretical results on the solution path estimators of Algorithm \ref{alg:first}. Our first result shows that $\tilde\bbf_t$'s and $\tilde\bl_i$'s have similar asymptotic results as those given in \cite{Bai2002Determining} and \cite{fan2013large}.

\begin{theorem}\label{th1}
Under Assumptions \ref{ass1}-\ref{ass5}, for $2\leq K< \infty$ in Algorithm \ref{alg:first},
\begin{eqnarray*}
\tilde{\bbf}_t&=&\tilde{\bW}_0\bbf_t^0+O_p\Big(\frac{\log{p}}{\sqrt{T}}+\frac{1}{\sqrt{p}}\Big),\\
\tilde{\bl}_i&=&\tilde{\bW}_0^{-1}\bl_i^0+O_p\Big(\frac{\log{p}}{\sqrt{T}}+\frac{1}{\sqrt{p}}\Big),\\
\tilde{\bl}_i^{\prime}\tilde{\bbf}_t&=&\bl_i^{0\prime}\bbf_t^0+O_p\Big(\frac{\log{p}}{\sqrt{T}}+\frac{1}{\sqrt{p}}\Big),\\
\tilde{\epsilon}_{it}&=:&y_{it}-\tilde{\bl}_i^{\prime}\tilde{\bbf}_t=\epsilon_{it}+O_p\Big(\frac{\log{p}}{\sqrt{T}}+\frac{1}{\sqrt{p}}\Big),
\end{eqnarray*}
where $\tilde{\bW}_0=\big\{\sum^p_{i=1}h_i(0)\tilde{\bl}_i\tilde{\bl}_i^{\prime}\big\}^{-1}\sum^p_{i=1}h_i(0)\tilde{\bl}_i\bl_i^{0\prime}$ satisfying $\tilde{\bW}_0\tilde{\bW}_0^{\prime}=\Ib_r$ with probability approaching one. If further ${p\log^2{p}}/{T}=o(1)$,
$$
\tilde{\bbf}_t=\tilde{\bW}_0\bbf_t^0+\frac{1}{2}\Big\{\sum^{p}_{i=1}h_i(0)\tilde{\bl}_i\tilde{\bl}_i^{\prime}\Big\}^{-1}\sum^p_{i=1}\tilde{\bl}_iD_{it}+o_p\Big(\frac{1}{\sqrt{p}}\Big),
$$
where $D_{it}=I(\epsilon_{it}\leq 0)-\tau$. If $\big(\log^2{p}\log^2{T}+{\log^3{T}}/{\sqrt{p}}\big)T/p+{\log^5{p}}/{\sqrt{T}}=o(1)$, there exists an $r\times r$ matrix $\bW$ satisfying $\bW\bW^{\prime}=\Ib_r$ with probability approaching one, such that
$$
\tilde{\bl}_i=\bW\bl_i^0+\frac{1}{2h_i(0)}\Big(\sum^T_{t=1}\tilde{\bbf}_t\tilde{\bbf}_t^{\prime}\Big)^{-1}\sum^T_{t=1}\tilde{\bbf}_tD_{it}+o_p\Big(\frac{1}{\sqrt{T}}\Big).
$$
\end{theorem}

Theorem \ref{th1} demonstrates that the computationally feasible factor and loading estimates match the realized factor and true loadings up to some orthogonal transformations, and recover the common components (factor returns) and idiosyncratic components (idiosyncratic returns) consistently for each variable. It also shows that the alternating iterative estimators share similar but slightly different asymptotics with the PCA-based estimators given in \cite{Bai2003Inferential} and \cite{fan2013large}. One reason is that our estimators rely on computing iterations. The other reason is the absence of an explicit decomposition of $\tilde{\bl}_i-\tilde{\bW}_0^{-1}\bl_i^0$ (or $\tilde{\bbf}_t-\tilde{\bW}_0\bbf_t^0$) in contrast to the eigen-decomposition of the PCA-based estimators. Indeed, the Bahadur representations present the principal correction terms of orders $p^{-1/2}$ and $T^{-1/2}$, but there aren't closed form expression for the $o_p(p^{-1/2})$ and $o_p(T^{-1/2})$ terms.

\textbf{Remark:}
Lemma 5 in the supplementary material demonstrates that the asymptotic results for $\tilde{\bl}_i$ and $\tilde{\bbf}_t$ in Theorem \ref{th1} can be strengthened to
\begin{eqnarray*}
\max_t\left\|\tilde{\bbf}_t-\bW^{(K)}\bbf_t^0\right\|&=&O_p\Big(\frac{\log{p}}{\sqrt{T}}+\frac{1}{\sqrt{p}}\Big),\\
\max_i\left\|\tilde{\bl}_i-(\bW^{(K)})^{-1}\bl_i^0\right\|&=&O_p\Big(\frac{\log{p}}{\sqrt{T}}+\frac{1}{\sqrt{p}}\Big),\\
\max_i\left\|\tilde{\bl}_i^{\prime}\tilde{\bbf}_t-\bl_i^{0\prime}\bbf_t^0\right\|&=&O_p\Big(\frac{\log{p}}{\sqrt{T}}+\frac{1}{\sqrt{p}}\Big),
\end{eqnarray*}
where $\bW^{(K)}$ is defined before Lemma 3 in the supplementary material . However, the rate for the common components are incorrect uniformly in $t$ except for assuming $\max_t\|\bbf_t^0\|=O_p(1)$ as in \cite{Chen2019Quantile}, which is far too restrictive.

{
In the above analysis, we assumed that the true number of factors $r$ is known in advance. However, in practice, $r$ is unknown and should be determined prior to implementing the robust iterative algorithm. In the remainder of this section, we introduce a robust method for determining the number of factors, which is of independent interest.
Our ``Robust Eigenvalue-Ratio" (RER) method is inspired by the ``Eigenvalue-Ratio" (ER) method in \cite{Ahn2013Eigenvalue}. That is,
\begin{equation}\label{equ:rer}
\hat r_{\text{RER}}=\bargmax_{1\le j\le r_{{\rm max}}-1}\left\{\frac{\lambda_j\Big(\tilde \bL(r_{\max})^\prime \tilde\bL(r_{\max})/p\Big)}{\lambda_{j+1}\Big(\tilde \bL(r_{\max})^\prime  \tilde\bL(r_{\max})/p\Big)}\right\},
\end{equation}
where $r_{\max}$ is a predetermined constant larger than $r$ and $\tilde \bL(r_{\max})$ is the estimated factor loading matrix  by the iterative algorithm in Algorithm \ref{alg:first} if we assume the number of factors is $r_\text{max}$. The notation $\lambda_j(\Ab)$ denotes the $j$-th largest eigenvalue of a nonnegative definitive matrix $\Ab$.

To analyze the theoretical properties of the estimator $\hat r_{\text{RER}}$, we assumed the following modification of Assumption {\bf 1'}. Let $\bL(m)$ be the factor loading matrix pretending that there are $m$ columns.

\noindent{\sc Assumption {\bf 1}''} {\it Assumption \ref{ass1} holds and}{
\begin{enumerate}
\item [(1)] \emph{the singular values of $\bL(r_{\max})^{\prime}\bH\bL(r_{\max})/p$ and $\bL(r_{\max})^\prime \bH\bL^0/p$ are bounded away from zero and infinity.}
    \item[(2)] {\it $\max_i\|\bl_i(r_{\max})\|\leq C$ for some generic constant $C$ where $\bl_i(r_{\max})$ is the $i$-th row of $\bL(r_{\max})$.}
\end{enumerate}}

The following theorem shows the property of $\tilde\bl_i(r_{\max})$ from the robust iterative algorithm  with $r_{\max}>r$, where $\tilde\bl_i(r_{\max})$ is the $i$-th row of $\tilde \bL(r_{\max})$.
\begin{theorem}\label{th2}
	Under Assumption {\bf 1''} and Assumptions \ref{ass2}-\ref{ass5}, for $1\le K<\infty$, there exists a positive definite matrices $\bW^{(K)}_{r_{\max}}$ of dimension $r_{\max}\times r_{\max}$ such that
	\[
	 \max_{i}\Big\|\tilde\bl_i(r_{\max})-\bW^{(K)}_{r_{\max}}(\bl_i^{0\prime},\zero^\prime)^\prime\Big\|=o_p(1).
	\]
\end{theorem}

Theorem \ref{th2} demonstrates that the leading $r$ eigenvalues of $\lambda_j\Big(\tilde \bL(r_{\max})^\prime \tilde\bL(r_{\max})/p\Big)$ are of order 1 while the remaining $(r_{\max}-r)$ eigenvalues are $o_p(1)$. Thus the eigenvalue ratio in (\ref{equ:rer}) is maximized asymptotically only at $j=r$ and consequently we  have the following theorem.
\begin{theorem}\label{coro:th2}
	Under Assumption {\bf 1''} and Assumptions \ref{ass2}-\ref{ass5},  for the estimator $\hat r_{\text{RER}}$, we have
 $$\Pr \big(\hat r_{\text{RER}}=r\big)\rightarrow 1, \ \ \text{as} \ \  p,T\rightarrow \infty.$$
	\end{theorem}

}

\section{Numerical experiments}\label{se}
\subsection{Data generating procedure}

In this section, we introduce the general Data Generating Procedures (DGPs), which are similar as those in the simulation studies of \cite{He2020Large}. In detail,
	\begin{align}\label{align:1}
	 &y_{it}=\sum\limits_{j=1}^{r}l_{ij}f_{jt}+\sqrt{\theta}u_{it},\quad u_{it}=\sqrt{\frac{1-\rho^2}{1+2J\beta^2}}e_{it}, \nonumber \\
&e_{it}=\rho e_{i,t-1}+(1-\beta)w_{it}+\sum_{l={\rm max}\{i-J,1\}}^{{\rm min}\{i+J,p\}}\beta w_{lt}, \ \ i=1,\ldots,p, \ \ t=1,\ldots,T,
	\end{align}
	where $\bw_t=(w_{1t},\ldots,w_{pt})^\top$ are  generated from  different distributions, the loadings $l_{ij}$'s are independently drawn from the standard normal distribution. In model (\ref{align:1}), $\rho$ controls the serial correlations of idiosyncratic errors, $\theta$ controls the signal to noise ratio (SNR), and the parameters $\beta$ and $J$ jointly control the cross-sectional correlations.

\subsection{Estimation of loading spaces, factor spaces and common components}\label{sec:3.2}
In this section, we assess the finite sample performances of the Robust Iterative Estimation Procedure (RIP) in terms of estimating loading spaces, factor spaces and common components. We compare the RIP with the Robust Two-Step (RTS) method proposed by \cite{He2020Large} and  the conventional PCA method. It is worth pointing out that the RTS method assumed that the common factors and idiosyncratic errors are jointly elliptically distributed. We consider the following two scenarios.

\vspace{0.5em}

	\textbf {Scenario A} Set $r=3,\theta=1,\rho=\beta=J=0$. We consider three cases on the joint distribution of $(\bbf_t^{\prime},\bw_t^{\prime})^{\prime}$: (i)  multivariate Gaussian distribution $\mathcal{N}(\zero,\Ib_{p+r})$; (ii) multivariate centralized $t$ distributions $t_{\nu}(\zero,\Ib_{p+r})$ with degree $\nu=3$; (iii) $\bbf_t$'s are generated from multivariate Gaussian distribution $\mathcal{N}(\zero,\Ib_{r})$ while all elements of $\bw_t$ are \emph{i.i.d.} samples from symmetric $\alpha$-stable distribution $S_{\alpha}(\beta,\gamma,\delta)$ with skewness parameter $\beta=0$, scale parameter $\gamma=1$ and location parameter $\delta=0$, $\alpha=1,1.5$. The combinations of  $(p,T)$ are set as $\big\{(150,100),(250,100),(250,150),(250,200)\big\}$.

\vspace{0.5em}

	\textbf {Scenario B} Set $r=3,\theta=0.5,\rho=0.2,\beta=0.2,J=3$. The settings on the joint distribution of $(\bbf_t^{\prime}, \bw_t^{\prime})^{\prime}$ are the same as those in Scenario A.

\vspace{0.5em}

In \textbf {Scenario A} (i) and (ii), the settings correspond to simple cases without any serial correlations of idiosyncratic errors and $(\bbf_t^\prime,\bw_t^\prime)^\prime$ are jointly from elliptical distributions.   $\mathcal{N}(\zero,\Ib_{p+r})$ satisfies the condtions for all three methods, while heavy-tailed $t_{3}(\zero,\Ib_{p+r})$ perfectly satisfies the  assumptions for RTS but not for PCA. In \textbf {Scenario A} (iii), the idiosyncratic errors are generated from $\alpha$-stable distributions which violates the conditions for both RTS and PCA. {In Scenario A, $w_{it}$'s are generated from symmetric distributions such that Assumption \ref{ass2} for RIP is satisfied with $\tau=1/2$. We also consider $\tau=0.75$ for RIP method in Scenario A, and in this case the panel observations $\{y_{it}\}$ are adjusted by $Q_{{3}/{4}}$, the third quartile of $u_{it}$. That's to say, the panel observations are  now \{$y_{it}-\sqrt{\theta}Q_{{3}/{4}}$\} such that Assumption \ref{ass2} for RIP is satisfied with $\tau=3/4$.}
In \textbf{Scenario B}, $(\bbf_t^{\prime},\bw_t^{\prime})^{\prime}$ are generated parallel to \textbf{Scenario A}, but the errors are now serially and cross-sectionally correlated by setting $\rho=0.2,\beta=0.2,J=3$. {We only consider the case $\tau=0.5$ for RIP in Scenario B since the theoretical value of $Q_{3/4}$ is not easy to compute for this case.} In all simulations, five initial values for RIP are tried and the minimum loss solution is left.

\begin{figure}[h]
 \centering
 \begin{minipage}[!t]{0.49\linewidth}
    \includegraphics[width=8.3cm, height=6.3cm]{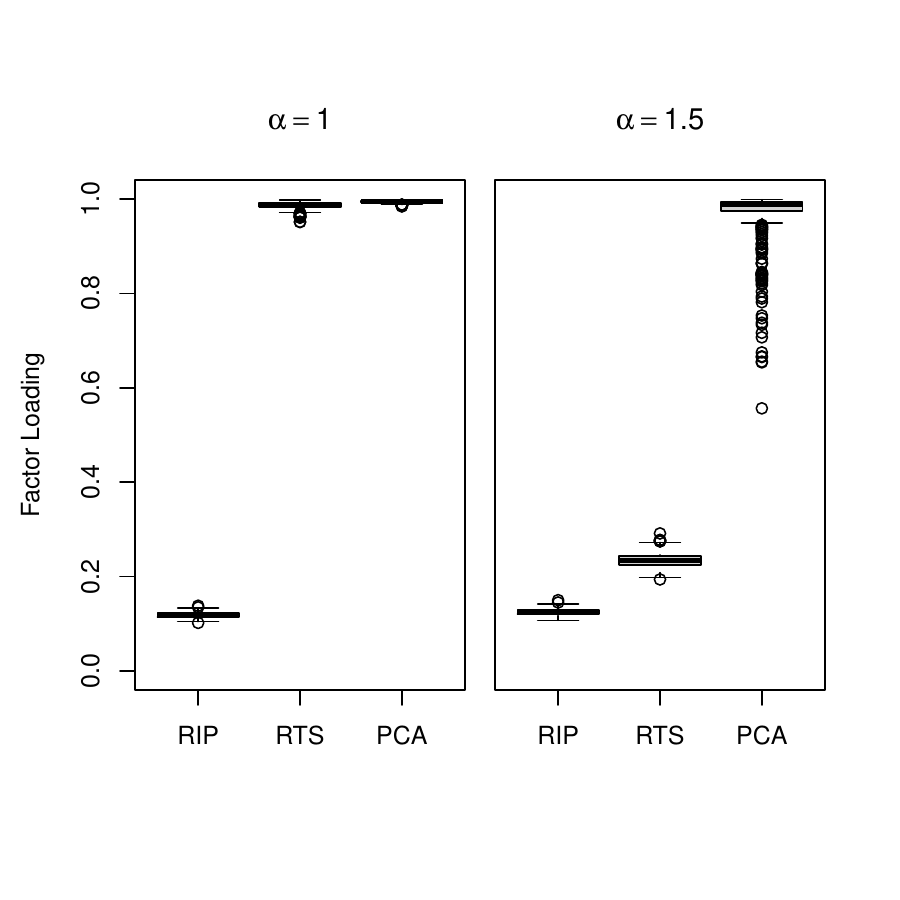}\\
  \end{minipage}
  \begin{minipage}[!t]{0.48\linewidth}
    \includegraphics[width=8.3cm,     height=6.3cm]{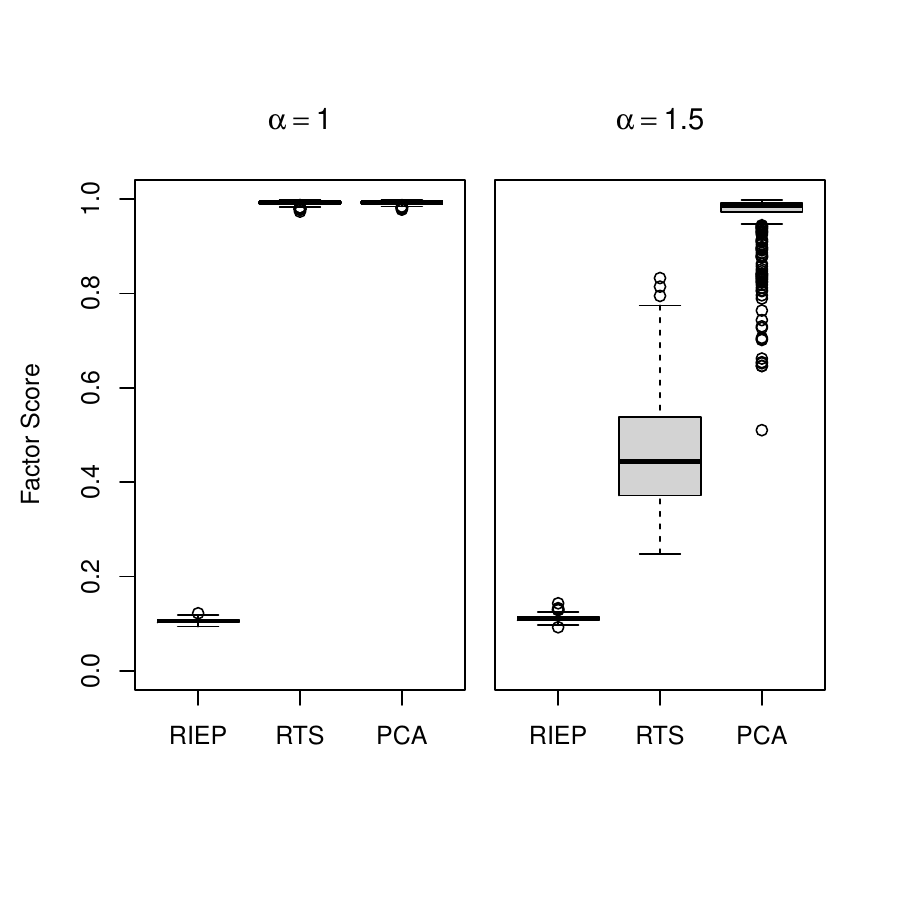}\\
  \end{minipage}
 \caption{Boxplots of the estimation errors of the estimated factor loadings and scores by RIP RTS and PCA methods  under symmetric $\alpha$-Stable distributions in Scenario A (iii) with $\alpha=1,1.5$.   $p=250,T=200$.}\label{fig:1}
 \end{figure}

To evaluate the empirical performances of different methods, we compare the measurement indices in \cite{He2020Large}, that is,  the \textbf{ME}dian of the normalized estimation \textbf{E}rrors for   \textbf{C}ommon \textbf{C}omponents in terms of the matrix Frobenius norm, denoted as MEE-CC; the \textbf{AV}erage estimation \textbf{E}rror for the \textbf{F}actor \textbf{L}oading matrices, denoted as AVE-FL; and the \textbf{AV}erage estimation \textbf{E}rror for the \textbf{F}actor \textbf{S}core matrices, denoted as AVE-FS. In detail,  the AVE-CC, AVE-FL and AVE-FS are defined as
\[
\begin{array}{ccl}
\text{MEE-CC}=\text{median}\left\{\|\hat\bL_m\hat{\bF}_m-\bL^0\bF^{0}\|_F^2/\|\bL^0\bF^{0}\|_F^2,m=1,\ldots,M\right\},\\
\text{AVE-FL}=\sum_{m=1}^M \cD(\hat \bL_m,\bL)/M, \hspace{0.5em} \text{and} \hspace{0.5em}
\text{AVE-FS}=\sum_{m=1}^M \cD(\hat \bF_m,\bF)/M,
\end{array}
\]
where $M$ is the number of replicates, $\hat \bL_m$ and $\hat \bF_m$ are respectively the estimators of the factor loading matrix and factor score matrix from the $m$-th replicate, and for {two column-wise orthogonal matrices} $\bQ_1$ and $\bQ_2$ of sizes $p\times q_1$ and $p\times q_2$,
\[
\cD(\bQ_1,\bQ_2)=\bigg(1-\frac{1}{\max{(q_1,q_2)}}\text{Tr}\Big(\bQ_1\bQ_1^{\prime}\bQ_2\bQ_2^{\prime}\Big)\bigg)^{1/2}.
\]
From the definition of $\cD(\bQ_1,\bQ_2)$, we can easily deduce that it is  a quantity between 0 and 1, which measures the distance between the column spaces of  $\bQ_1$ and $\bQ_2$. $\cD(\bQ_1,\bQ_2)=0$ indicates the column spaces of $\bQ_1$ and $\bQ_2$  are the same, while  $\cD(\bQ_1,\bQ_2)=1$ indicates the column spaces of $\bQ_1$ and $\bQ_2$ are orthogonal. In fact, $\cD(\cdot,\cdot)$  particularly fits to quantify  the accuracy of estimated factor loading/score matrices as they are not separately identifiable. All the simulation results are based on $M=500$ replicates.

\begin{table}[!h]
  \caption{Simulation results for Scenario A, the values in the parentheses are the interquartile ranges for MEE-CC and standard deviations for AVE-FL and AVE-FS.}
  \label{tab:1}
  \renewcommand{\arraystretch}{1}
  \centering
  \selectfont
  \begin{threeparttable}
   \scalebox{0.75}{\begin{tabular*}{19.8cm}{ccccccccccccccccccccccccccccc}
\toprule[2pt]
&\multirow{2}{*}{Type}&\multirow{2}{*}{Method}  &\multicolumn{3}{c}{$(p,T)=(150,100)$}&\multicolumn{3}{c}{$(p,T)=(250,100)$} \cr
\cmidrule(lr){4-6} \cmidrule(lr){7-9}
&&                 &$\text{MEE\_CC}$     &$\text{AVE\_FL}$      &$\text{AVE\_FS}$      &$\text{MEE\_CC}$      &$\text{AVE\_FL}$      &$\text{AVE\_FS}$  \\
\midrule[1pt]
&$\mathcal{N}(\zero,\Ib_{p+m})$   &RIP($\tau=0.5$)    &0.03(0.00)	&0.13(0.01)	&0.10(0.01)	 &0.02(0.00)	 &0.13(0.01)	&0.08(0.01) \\
&                               &RIP($\tau=0.75$)    &0.03(0.01)	    &0.14(0.02)	&0.11(0.01) &0.03(0.00) &0.14(0.01)	 &0.09(0.01)	\\
&                                 &RTS     &0.02(0.00)	&0.11(0.01)	&0.08(0.01)	 &0.01(0.00)	 &0.11(0.01)	&0.06(0.00)  \\
&                                 &PCA     &0.02(0.00)	&0.10(0.01)	&0.08(0.01)	 &0.01(0.00)	 &0.10(0.01)	&0.06(0.00)  \\
\cmidrule(lr){4-9}
&$t_{3}(\zero,\Ib_{p+m})$         &RIP($\tau=0.5$)    &0.03(0.01)	&0.16(0.03)	&0.11(0.02)	 &0.03(0.01)	 &0.16(0.02)	&0.10(0.03) \\
&                               &RIP($\tau=0.75$)    &0.05(0.02)	    &0.18(0.04)	 &0.12(0.03) &0.04(0.01)	 &0.18(0.04)	&0.11(0.01) \\
&                                 &RTS     &0.02(0.00)	&0.12(0.01)	&0.08(0.01)	 &0.02(0.00)	 &0.11(0.01)	&0.07(0.01)  \\
&                                 &PCA     &0.04(0.03)	&0.20(0.06)	&0.10(0.03)	 &0.04(0.03)	 &0.20(0.06)	&0.09(0.05)  \\

\cmidrule(lr){4-9}
&$S_{1}(0,1,0)$                 &RIP($\tau=0.5$)     &0.05(0.01)	    &0.18(0.01)	 &0.14(0.01) &0.04(0.01)	 &0.17(0.01)	&0.11(0.01) \\
&                               &RIP($\tau=0.75$)    &0.12(0.02)	    &0.26(0.02)	 &0.22(0.02) &0.10(0.01)	 &0.26(0.02)	&0.17(0.03) \\
&                               &RTS     &661.3(3510.92)	&0.98(0.01)	 &0.98(0.01) &894.64(3855.27)  	 &0.99(0.01)	&0.98(0.01)  \\
&                               &PCA     &6404.27(43169.89)	&0.99(0.00)	 &0.98(0.01) &11876.97(64758.55)	 &0.99(0.00)	&0.99(0.01)  \\
\cmidrule(lr){4-9}
&$S_{1.5}(0,1,0)$                &RIP($\tau=0.5$)    &0.05(0.01)	 &0.18(0.01)	&0.15(0.01)	 &0.04(0.01)     &0.18(0.01)	&0.11(0.01) \\
&                               &RIP($\tau=0.75$)    &0.08(0.01)	    &0.22(0.01)	 &0.18(0.01) &0.06(0.01)	 &0.21(0.01)	&0.14(0.01) \\
&                                 &RTS   &0.51(0.69)	 &0.32(0.03)	&0.51(0.11)	 &0.43(0.53)   &0.33(0.03)	&0.46(0.13)  \\
&                                 &PCA   &8.21(15.96)	 &0.92(0.08)	&0.91(0.09)	 &9.15(20.44)	&0.94(0.08)	&0.93(0.09)  \\
\midrule[1pt]
&\multirow{2}{*}{Type}&\multirow{2}{*}{Method}  &\multicolumn{3}{c}{$(p,T)=(250,150)$}&\multicolumn{3}{c}{$(p,T)=(250,200)$} \cr
\cmidrule(lr){4-6} \cmidrule(lr){7-9}
&&                   &$\text{MEE\_CC}$     &$\text{AVE\_FL}$      &$\text{AVE\_FS}$      &$\text{MEE\_CC}$      &$\text{AVE\_FL}$      &$\text{AVE\_FS}$  \\
\midrule[1pt]
&$\mathcal{N}(\zero,\Ib_{p+m})$   &RIP($\tau=0.5$)   &0.02(0.00)	&0.10(0.01)	&0.08(0.00)	 &0.01(0.00)	 &0.09(0.00)	&0.08(0.00) \\
&                               &RIP($\tau=0.75$)    &0.02(0.00)	    &0.11(0.01)	 &0.09(0.00) &0.02(0.00)	 &0.10(0.05)	&0.09(0.05) \\
&                                 &RTS     &0.01(0.00)	&0.08(0.00)	&0.06(0.00)	 &0.01(0.00)	 &0.07(0.00)	&0.06(0.00)  \\
&                                 &PCA     &0.01(0.00)	&0.08(0.00)	&0.06(0.00)	 &0.01(0.00)	 &0.07(0.00)	&0.06(0.00)  \\
\cmidrule(lr){4-9}
&$t_{3}(\zero,\Ib_{p+m})$         &RIP($\tau=0.5$)    &0.02(0.00)	&0.13(0.02)	&0.08(0.01)	 &0.02(0.00)	 &0.11(0.01)	&0.08(0.01) \\
&                               &RIP($\tau=0.75$)    &0.03(0.01)	    &0.15(0.03)	 &0.09(0.03) &0.02(0.01)	 &0.13(0.04)	&0.09(0.04) \\
&                                 &RTS     &0.01(0.00)	&0.09(0.00)	&0.06(0.01)	 &0.01(0.00)	 &0.08(0.00)	&0.06(0.01)  \\
&                                 &PCA     &0.03(0.02)	&0.17(0.05)	&0.08(0.03)	 &0.03(0.02)	 &0.16(0.05)	&0.08(0.02)  \\

\cmidrule(lr){4-9}
&$S_{1}(0,1,0)$                 &RIP($\tau=0.5$)    &0.03(0.00)	        &0.14(0.01)	 &0.11(0.01)	 &0.03(0.00)	    &0.12(0.01)	&0.11(0.01)  \\
&                               &RIP($\tau=0.75$)    &0.07(0.01)	    &0.21(0.03)	 &0.17(0.03) &0.06(0.01)	 &0.19(0.03)	&0.17(0.03) \\
&                               &RTS    &767.12(2878.55)	&0.99(0.01)	 &0.99(0.00)     &1058.58(4220.23)	 &0.99(0.01)	&0.99(0.00)  \\
&                               &PCA    &14874.71(75287.90)	&0.99(0.00)	 &0.99(0.00)	 &20822.62(104666.86)	 &0.99(0.00) &0.99(0.00)  \\
\cmidrule(lr){4-9}
&$S_{1.5}(0,1,0)$                &RIP($\tau=0.5$)   &0.03(0.00) &0.14(0.01)	&0.11(0.01)	       &0.03(0.00)  &0.13(0.01)	&0.11(0.01) \\
&                               &RIP($\tau=0.75$)    &0.05(0.01)	    &0.17(0.01)	 &0.13(0.01) &0.04(0.01)	 &0.15(0.01)	&0.13(0.01) \\
&                                &RTS   &0.37(0.45) &0.27(0.02)	&0.46(0.12)  &0.34(0.40) &0.23(0.01)	 &0.46(0.12)  \\
&                                &PCA   &11.31(23.41) &0.95(0.07) &0.95(0.07) 	 &11.75(25.33)	&0.97(0.06)	 &0.96(0.06)  \\
\bottomrule[2pt]
  \end{tabular*}}
  \end{threeparttable}
\end{table}

\begin{table}[!h]
  \caption{Simulation results for Scenario B, the values in the parentheses are the interquartile ranges for MEE-CC and standard deviations for AVE-FL and AVE-FS.}
  \label{tab:2}
  \renewcommand{\arraystretch}{1}
  \centering
  \selectfont
  \begin{threeparttable}
   \scalebox{0.77}{ \begin{tabular*}{18.5cm}{ccccccccccccccccccccccccccccc}
\toprule[2pt]
&\multirow{2}{*}{Type}&\multirow{2}{*}{Method}  &\multicolumn{3}{c}{$(p,T)=(150,100)$}&\multicolumn{3}{c}{$(p,T)=(250,100)$} \cr
\cmidrule(lr){4-6} \cmidrule(lr){7-9}
&&                 &$\text{MEE\_CC}$     &$\text{AVE\_FL}$      &$\text{AVE\_FS}$      &$\text{MEE\_CC}$      &$\text{AVE\_FL}$      &$\text{AVE\_FS}$  \\
\midrule[1pt]
&$\mathcal{N}(\zero,\Ib_{p+m})$   &RIP     &0.01(0.00)	&0.09(0.01)	&0.07(0.01)	 &0.01(0.00)	 &0.09(0.01)	&0.06(0.00)\\
&                                 &RTS      &0.01(0.00)	&0.08(0.01)	&0.06(0.00)	 &0.01(0.00)	 &0.07(0.00)	&0.05(0.00) \\
&                                 &PCA      &0.01(0.00)	&0.07(0.01)	&0.06(0.00)	 &0.01(0.00)	 &0.07(0.00)	&0.05(0.00) \\
\cmidrule(lr){4-9}
&$t_{3}(\zero,\Ib_{p+m})$         &RIP     &0.02(0.00)	&0.11(0.02)	&0.08(0.01)	 &0.01(0.00)	 &0.11(0.02)	&0.06(0.01)   \\
&                                 &RTS      &0.01(0.00)	&0.08(0.01)	&0.06(0.01)	 &0.01(0.00)	 &0.08(0.01)	&0.05(0.01)   \\
&                                 &PCA      &0.02(0.01)	&0.14(0.05)	&0.07(0.03)	 &0.02(0.01)	 &0.14(0.05)	&0.05(0.03)   \\
\cmidrule(lr){4-9}
&$S_{1}(0,1,0)$                   &RIP     &1.45(181.64)	&0.51(0.15)	&0.48(0.17)	 &0.17(0.05)	&0.35(0.09)	&0.26(0.11)   \\
&                                 &RTS      &2293.52(18769.22) 	&0.99(0.01)	&0.98(0.01)	 &4349.65(28488.78) 	&0.99(0.00)	&0.98(0.01)   \\
&                                 &PCA      &3202.13(21579.74)	&0.99(0.00)	&0.98(0.01)	 &5950.76(32357.06) 	&0.99(0.00)	&0.99(0.01)   \\
\cmidrule(lr){4-9}
&$S_{1.5}(0,1,0)$                 &RIP     &0.04(0.01)	&0.16(0.01)	&0.13(0.01)	  &0.03(0.01)	 &0.16(0.01)	&0.10(0.01)\\
&                                 &RTS      &0.34(0.54)	&0.28(0.05)	&0.45(0.13)	 &0.26(0.51)	 &0.28(0.04)	&0.41(0.15)\\
&                                 &PCA      &4.40(8.39)	&0.83(0.14)	&0.82(0.14)	 &4.92(10.51)	 &0.86(0.14)	&0.84(0.15)\\
\midrule[1pt]
&\multirow{2}{*}{Type}&\multirow{2}{*}{Method}  &\multicolumn{3}{c}{$(p,T)=(250,150)$}&\multicolumn{3}{c}{$(p,T)=(250,200)$} \cr
\cmidrule(lr){4-6} \cmidrule(lr){7-9}
&&                   &$\text{MEE\_CC}$     &$\text{AVE\_FL}$      &$\text{AVE\_FS}$      &$\text{MEE\_CC}$      &$\text{AVE\_FL}$      &$\text{AVE\_FS}$  \\
\midrule[1pt]
&$\mathcal{N}(\zero,\Ib_{p+m})$   &RIP     &0.01(0.00)	&0.07(0.00)	&0.06(0.00)	 &0.01(0.00)	 &0.06(0.00)	&0.06(0.00) \\
&                                 &RTS      &0.01(0.00)	&0.06(0.00)	&0.05(0.00)	 &0.00(0.00)	 &0.05(0.00)	&0.05(0.00) \\
&                                 &PCA      &0.01(0.00)	&0.06(0.00)	&0.05(0.00)	 &0.00(0.00)	 &0.05(0.00)	&0.05(0.00) \\
\cmidrule(lr){4-9}
&$t_{3}(\zero,\Ib_{p+m})$         &RIP     &0.01(0.00)	&0.09(0.01)	&0.06(0.01)	 &0.01(0.00)	 &0.08(0.01)	&0.06(0.01)   \\
&                                 &RTS      &0.01(0.00)	&0.07(0.00)	&0.05(0.01)	 &0.01(0.00)	 &0.06(0.00)	&0.05(0.00)   \\
&                                 &PCA      &0.01(0.01)	&0.12(0.03)	&0.05(0.01)	 &0.01(0.01)	 &0.11(0.04)	&0.05(0.01)   \\
\cmidrule(lr){4-9}
&$S_{1}(0,1,0) $                  &RIP     &0.11(0.01)	&0.26(0.05)	&0.22(0.05)	 &0.09(0.02)	 &0.22(0.03)	&0.21(0.03)   \\
&                                 &RTS     &5256.14(31351.96) 	&0.99(0.00)	&0.99(0.00)	 & 6594.49(41148.65)	&0.99(0.00)	&0.99(0.00)   \\
&                                 &PCA     &7429.45(37392.76)	&0.99(0.00)	&0.99(0.00)	 &10413.65(51591.05)	&0.99(0.00)	&0.99(0.00)   \\
\cmidrule(lr){4-9}
&$S_{1.5}(0,1,0)$                 &RIP     	 &0.03(0.00)	 &0.13(0.01)	&0.10(0.01) &0.02(0.00)	 &0.11(0.01)	&0.10(0.01) \\
&                                 &RTS     &0.23(0.37)	&0.22(0.02)	&0.40(0.13) 	 &0.20(0.31)	&0.19(0.02)	&0.39(0.13) \\
&                                 &PCA      &5.99(11.94)	&0.88(0.13)	&0.87(0.14)	 &6.27(12.92)	&0.90(0.12)	&0.90(0.12) \\
\bottomrule[2pt]
  \end{tabular*}}
  \end{threeparttable}
\end{table}

The simulation results for Scenario A and  Scenario B  are reported in Table \ref{tab:1} and Table \ref{tab:2}, respectively. From Table \ref{tab:1}, we see that for multivariate Gaussian case in Scenario A (i), all three methods perform very well while the RIP seems a bit worse. This is expected as in the regression that least absolute regression is less efficient than least square regression when errors are normal. For multivariate $t$ distribution with degree of freedom 3 in Scenario A (ii), RTS performs the best as the elliptical assumption is satisfied. The RIP performs satisfactorily though not as well as RTS. The PCA is the worst, which reflects the effect of the non-existence of the forth moment. The advantages of the proposed RIP are well illustrated in Scenario A (iii), where the errors are from symmetric $\alpha$-stable distribution. Figure \ref{fig:1} shows the boxplots of the estimation errors of the estimated factor loadings and scores by RIP, RTS and PCA methods over 500 replications, with  $\alpha=1$, $\alpha=1.5$ and $p=250,T=200$. From Figure \ref{fig:1}, we see that the RIP performs very well while the RTS and PCA totally lose power. This is expected since neither the elliptical assumption nor the forth moment condition is satisfied.
From Table \ref{tab:1}, it can also be concluded that the performances of the RIP tend to be better as $T$ and/or $p$ increase which is consistent with the theoretical results. In summary, RIP is quite stable, but RTS and PCA become worse substantially as the tail becomes thicker. {At last, we see that the RIP performs comparably for $\tau=0.5$ and $\tau=0.75$}.

Next, we turn to Scenario B when both cross-sectional and serial correlations are present. The superiority of the RIP over the RTS and PCA is clearly illustrated when the idiosyncratic errors are from $\alpha$-stable distribution. For $\alpha=1$, when $T,p$ are small, the RIP does not perform  well, though far much better than the RTS and PCA. As $T, p$ grow large, the performance of RIP boosts, while RTS and PCA still does not work. When $\alpha=1.5$, $(T,p)=(100,150)$ is enough to guarantee the good performance of RIP, while even when $T=200$ and $p=250$ the RTS and PCA still fall far behind. In summary, the proposed RIP method performs robustly in both light-tailed and heavy-tailed settings.

\subsection{Selection of the number of factors}

In this section, we assess the finite sample performance of the proposed ``Robust Eigenvalue-Ratio" method (RER) for factor number selection. We compare our RER method with the ``Eigenvalue-Ratio" (ER) method in \cite{Ahn2013Eigenvalue}, the ``Multivariate-Kendall's tau-Eigenvalue-Ratio" (MKER) method in \cite{yu2019robust} and  the classical ``Information Criteria" (IC) method in \cite{Bai2002Determining}. To evaluate the empirical performance of different methods, we consider the following scenario.

\begin{table}[!h]
\caption{Simulation results in the form $x(y|z)$ for Scenario C, $x$ is the sample mean of the estimated factor numbers based on 200 replications, $y$ and $z$ are the numbers of underestimation and overestimation, respectively.}\label{tab:3}
  \renewcommand{\arraystretch}{1}
  \centering
  \selectfont
  \begin{threeparttable}
   \scalebox{0.85}{ \begin{tabular*}{18cm}{ccccccccccccccccccccccccccccc}
\toprule[2pt]
Type&$p$&$T$&$r$&&RER($\tau=0.5$)&RER($\tau=0.75$)&IC&ER&MKER\\\hline
       	    \multirow{5}*{$\mathcal{N}(\zero,\Ib_{p+m})$}&50&50&3&&3.000(0$|$0)&3.000(0$|$0)&3.000(0$|$0)&3.000(0$|$0)&3.000(0$|$0)\\
       	    &100&100&3&&3.000(0$|$0)&3.000(0$|$0)&3.000(0$|$0)&3.000(0$|$0)&3.000(0$|$0)\\
       	    &150&150&3&&3.000(0$|$0)&3.000(0$|$0)&3.000(0$|$0)&3.000(0$|$0)&3.000(0$|$0)\\
       	    &200&200&3&&3.000(0$|$0)&3.000(0$|$0)&3.000(0$|$0)&3.000(0$|$0)&3.000(0$|$0)\\
       		\hline
       	    \multirow{5}*{$t_{3}(\zero,\Ib_{p+m})$}&50&50&3&&2.800(30$|$8)&2.760(33$|$6)&5.565(0$|$184)&2.770(31$|$7)&3.000(0$|$0)\\
       	    &100&100&3&&3.055(4$|$15)&3.005(7$|$12)&5.860(0$|$192)&3.030(6$|$13)&3.000(0$|$0)\\
       	    &150&150&3&&2.980(8$|$10)&2.935(9$|$3)&6.210(0$|$194)&2.960(9$|$9)&3.000(0$|$0)\\
       	    &200&200&3&&3.020(4$|$11)&2.980(4$|$3)&6.440(0$|$196)&3.005(4$|$8)&3.000(0$|$0)\\
       		\hline
       	    \multirow{5}*{$t_2$}&50&50&3&&2.975(5$|$1)&2.925(26$|$17)&4.935(5$|$159)&2.435(114$|$48)&2.600(60$|$2)\\
       	    &100&100&3&&3.000(0$|$0)&3.000(0$|$0)&5.640(0$|$189)&2.990(93$|$79)&2.980(4$|$1)\\
       	    &150&150&3&&3.000(0$|$0)&3.000(0$|$0)&6.065(0$|$194)&3.415(71$|$100)&3.000(0$|$0)\\
       	    &200&200&3&&3.000(0$|$0)&3.000(0$|$0)&6.555(0$|$197)&3.485(76$|$100)&3.000(0$|$0)\\
       		\hline
            \multirow{5}*{$S_{1}(0,1,0)$}&50&50&3&&2.500(114$|$47)&1.735(169$|$16)&7.790(1$|$198)&2.050(142$|$25)&2.280(137$|$39)\\
       	    &100&100&3&&2.990(2$|$1)&2.885(83$|$90)&8.000(0$|$200)&2.180(144$|$36)&2.325(129$|$36)\\
       	    &150&150&3&&2.990(1$|$0)&3.155(11$|$48)&8.000(0$|$200)&2.130(147$|$32)&2.285(144$|$38)\\
       	    &200&200&3&&3.000(0$|$0)&3.010(3$|$8)&8.000(0$|$200)&2.210(134$|$33)&2.200(151$|$35)\\
       		\hline
       	    \multirow{5}*{$S_{1.5}(0,1,0)$ }&50&50&3&&2.900(16$|$0)&2.725(45$|$10)&6.425(4$|$184)&2.030(155$|$28)&2.365(96$|$10)\\
       	    &100&100&3&&3.000(0$|$0)&3.000(0$|$0)&7.485(0$|$199)&2.085(146$|$29)&2.940(14$|$6)\\
       	    &150&150&3&&3.000(0$|$0)&3.000(0$|$0)&7.745(0$|$200)&1.805(164$|$17)&3.020(0$|$3)\\
       	    &200&200&3&&3.000(0$|$0)&3.000(0$|$0)&7.950(0$|$200)&1.885(158$|$22)&3.000(0$|$0)\\
\bottomrule[2pt]
  \end{tabular*}}
  \end{threeparttable}
\end{table}

\textbf {Scenario C}  Set $r=3,\theta=1,\rho=\beta=J=0$. We consider three cases on the joint distribution of $(\bbf_t^{\prime},\bw_t^{\prime})^{\prime}$: (i)  multivariate Gaussian distribution $\mathcal{N}(\zero,\Ib_{p+r})$; (ii) multivariate centralized $t$ distributions $t_{\nu}(\zero,\Ib_{p+r})$ with degree $\nu=3$; (iii)  $\bbf_t$'s are generated from multivariate Gaussian distribution $\mathcal{N}(\zero,\Ib_{r})$ while all elements of $\bw_t$ are \emph{i.i.d.} samples from $t_2$ distribution. (iv) $\bbf_t$'s are generated from multivariate Gaussian distribution $\mathcal{N}(\zero,\Ib_{r})$ while all elements of $\bw_t$ are \emph{i.i.d.} samples from symmetric $\alpha$-Stable distribution $S_{\alpha}(\beta,\gamma,\delta)$ with skewness parameter $\beta=0$, scale parameter $\gamma=1$ and location parameter $\delta=0$, $\alpha=1$ and $\alpha=1.5$. $(p,T)=\big\{(50,50),(100,100),(150,150),(200,200)\big\}$.
\vspace{0.5em}

{In Scenario C, we consider the cases $\tau=0.5$ and $\tau=0.75$ for RER method. For the case $\tau=0.75$, the panel observations are similarly adjusted as we did in Section \ref{sec:3.2}.}
In Table \ref{tab:3}, we show the simulation results  in the form $x(y|z)$,  where  $x$ is the sample mean of the estimated factor numbers over 200 replications, while $y$ and $z$ are the numbers of underestimation and overestimation, respectively. {Firstly, we can see that the RER method performs comparably for the cases $\tau=0.5$ and $\tau=0.75$, and thus we simply refer to RER ($\tau=0.5$) and RER($\tau=0.75$) as RER hereafter.}
For the light-tailed Gaussian case in Scenario C (i), all methods perform very well and $T=p=50$ is sufficient for guaranteeing a satisfactory performance. For the heavy-tailed cases in Scenario C (ii), (iii) and (iv), the classical IC method always overestimates the factor number by a large margin. The performances of ER method is barely satisfactory for $\alpha$-stable idiosyncratic errors, and it underestimates the factor numbers by a margin even $p=T=200$. It seems that the MKER and RER methods  perform the best in the heavy-tailed cases. For Scenario C (ii), the factors and the idiosyncratic errors are jointly $t_3$ distributed, thus MKER performs the best as it's specifically designed for this setting.
For Scenario C  (iii) and (iv), the factors are from Gaussian and the idiosyncratic errors are either from $t_2$ distribution or $\alpha$-stable distribution, the RER performs satisfactorily. It can be seen that as $p,T$ grow, the estimates by RER converge to the true factor numbers. Noticeably, the RER method performs much better than MKER for Scenario C (iv), especially when $\alpha=1$.

\section{An empirical study}\label{em}

We collected  the weekly share returns of the Standard $\&$ Poor 100 companies during the period between January 1st, 2018 and December 31st, 2019. The data set is available at \texttt{https://github.com/heyongstat/RIP}. By preliminary time series analysis techniques such as the Augmented Dickey-Fuller tests and sample auto-correlation functions, we found that all weekly return series are stationary. The kurtosis of a portion of series are much larger than 9, the theoretical kurtosis of $t_5$ distribution.
We centralized the log returns for further analysis.


We first consider the out-of-sample performance of the PCA, RTS and RIP, which are motivated by \cite{Kelly2019Characteristics}. At the end of each month $t$, the latest 52 weekly returns on and before $t$ are selected to train the factor model. The number of factors is recursively estimated in each rolling manipulation. For the PCA, RTS and RIP,  we estimate the number of factors by ``ER", ``MKER" and ``RER", respectively. With the estimated number of factors, we separately use the PCA, RTS and RIP to estimate the factor loading matrix and obtain $\hat{\bL}_t=(\hat{\bl}_{1,t},\ldots,\hat{\bl}_{p,t})^T$.
 To estimate the factors at time point $t+1$, we consider  the following cross-section regression model,
 \[
 y_{j,t+1}=\hat{\bl}_{j,t}^T\bbf_{t+1}+\eta_{j,t+1},\quad j=1,\ldots,100, \ t=1,..., 105,
 \]
 where $y_{j,t}$ is the centralized log return of company $j$ at week $t$, and $\eta_{jt}$ is the random error term. For the PCA and RTS, $\hat{\bbf}_{t+1}$ is obtained by the least square estimation, while for the RIP, it is obtained by least absolute regression. The ``Square Total $R^2$" is defined as
\[
\text{Square Total} \  R^2=1-\sum_{j=1}^{100}\sum_{t=52}^{104}(y_{j,t+1}-\hat{\bl}_{j,t}^T\hat{\bbf}_{t+1})^2(\sum_{j=1}^{100}\sum_{t=52}^{104}y_{j,t+1}^2)^{-1},
\]
and the "Absolute Total $R^2$" is defined as
\[
\text{Absolute Total} \  R^2=1-\sum_{j=1}^{100}\sum_{t=52}^{104}|y_{j,t+1}-\hat{\bl}_{j,t}^T\hat{\bbf}_{t+1}|(\sum_{j=1}^{100}\sum_{t=52}^{104}|y_{j,t+1}|)^{-1}.
\]

Our computation shows that the ``Square Total $R^2$"s for the PCA, RTS and RIP are $0.251$,  $0.255$ and  $0.243$, respectively. They are comparable.
The ``Absolute Total $R^2$"s for the PCA, RTS and RIP are 0.190,  0.192 and  0.197 respectively, which are also more or less the same. This concludes that as a safe replacement, the RIP achieves similar prediction power as the PCA and RTS.

We next compare the annual return of the year 2019 by constructing risk-minimization portfolios. In detail, under the framework of elliptical distribution,
denote the true scatter matrix of the returns by $\bSigma$, then the optimal risk-minimization portfolio weights are $\bm{\omega}_{opt}=\bSigma^{-1}\bm{1}/(\bm{1}^T\bSigma^{-1}\bm{1})$, where $\bm{1}$ is a vector of ones, see for example \cite{Chamberlain1983} and \cite{Owen2012on}. At week $t$, the  data of past 52 weeks are used to train the AFM. We denote the estimated common components and
idiosyncratic errors  as $\widehat{\mathbf{\mathcal{X}}}_t$ and $\widehat{\mathbf{\mathcal{E}}}_t$, respectively. We empirically estimate the scatter matrix as
\[
\widehat{\bSigma}_t=\frac{1}{52}\widehat{\mathbf{\mathcal{X}}}_t^T\widehat{\mathbf{\mathcal{X}}}_t+\text{HarTh}(\frac{1}{52}\widehat{\mathbf{\mathcal{E}}}_t^T\widehat{\mathbf{\mathcal{E}}}_t),
\]
where $\text{HarTh}(\cdot)$ is the hardthresholding operator defined in \cite{Beckel2008Cov} simply to guarantee the invertibility of $\widehat{\bSigma}_t$.

The portfolio weights are thus specified as {$\widehat{\bm{\omega}}_t=\widehat{\bSigma}^{-1}_t\bm{1}/(\bm{1}^T\widehat{\bSigma}^{-1}_t\bm{1})$}
and the return of the risk-minimization portfolio strategy at week $t$ is  $\widehat{\bm{\omega}}^T_t {\bm x}_t$, where ${\bm x}_t$  is composed of the corresponding returns at week $t$. In the left panel of Figure \ref{fig:combine}, we display the net value curves of the risk-minimization portfolios during the year 2019. It shows that the empirical RIP portfolio leads to the highest annual return, and that RTS takes the second place while PCA lies at the bottom.

\begin{figure}[h]
 \centering
 \begin{minipage}[!t]{0.48\linewidth}
    \includegraphics[width=1\textwidth]{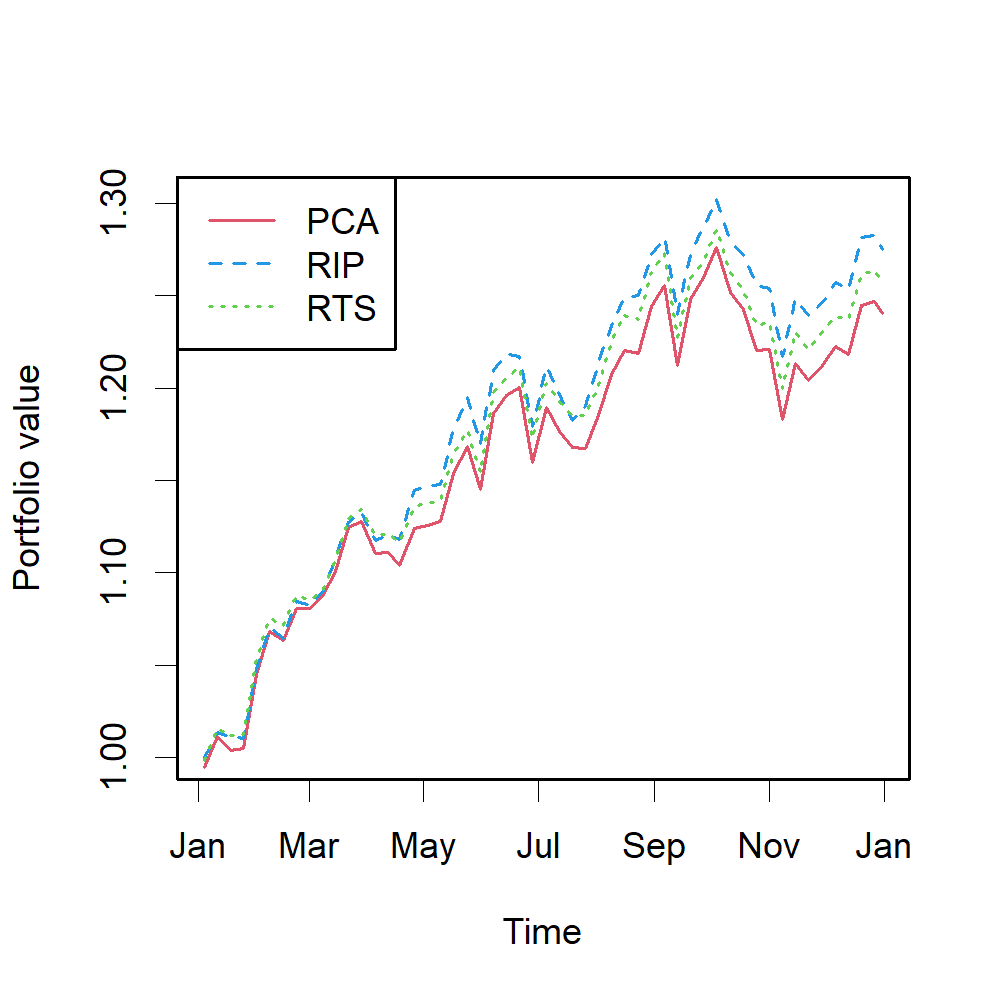}\\
  \end{minipage}
  \begin{minipage}[!t]{0.45\linewidth}
    \includegraphics[width=1\textwidth]{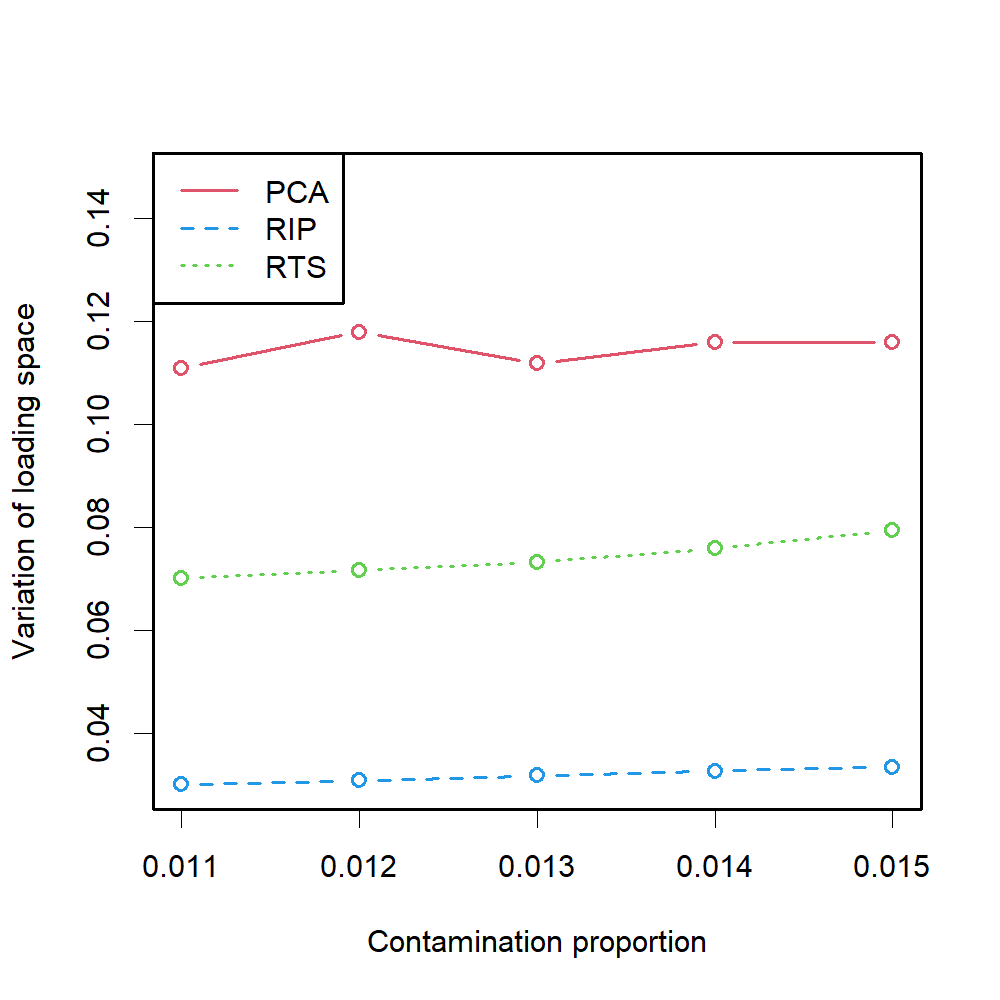}\\
  \end{minipage}
 \caption{The left panel shows net value curves of empirical portfolios. The right panel shows the average variation of estimated loading space in $100$ replications with growing proportion of outliers.}\label{fig:combine}
 \end{figure}

To investigate the robustness of various methods, we assess their sensitivity to outliers. All methods estimate $r=1$ for the whole sample. We randomly select a proportion of the demeaned log returns and multiply them by 5, and evaluate the sensitivity  by the variation of the estimated loading space compared with the original estimated space for each method.  The random contamination procedure above were repeated  $100$ times. We report the mean variation for a variety of contamination proportions in the right panel of Figure \ref{fig:combine}. It is clear that the RIP always has smallest variation. One can also tell that the RTS is more robust compared with PCA.

\section{Conclusion and Discussion}\label{condis}

In this paper, we presented a way to do robust factor analysis without any moment constraint. The method relies on alternating the quantile regressions in the factors cross-sectionally and in the loadings temporally. We show that after several iterations, the terminated solution can not only identify the common components but also estimate the factors and scores consistently up to some orthogonal transformations. This provides at least a safe replacement of the PCA-based factor analysis when there are heavy-tailed idiosyncratic errors. There are still some problems that are eager to be solved in the future research. First, is there theoretical guarantee that the $\tilde{\bl}_i$'s and $\tilde{\bbf}_t$'s will converge to $\hat{\bl}_i$'s and $\hat{\bbf}_t$'s as $K\rightarrow\infty$? Thus the computing error for $C^{(K)}_{ij}$ can be theoretically controlled. This is so far difficult to achieve or prove, and we leave it to our future research work. Second, one can extend the current work to a general class of loss functions beyond the weighted absolute deviation loss.
\section{Acknowledgement}
He's work is supported by  National Science Foundation (NSF) of  China (12171282,11801316), National Statistical Scientific Research Key Project (2021LZ09), Young Scholars Program of Shandong University, Project funded by
China Postdoctoral Science Foundation (2021M701997) and the Fundamental Research Funds of Shandong University. Kong's work is partially supported by NSF China (71971118 and 11831008) and the WRJH-QNBJ Project and Qinglan Project of Jiangsu Province.  The authors would like to thank professor Xinsheng Zhang at Fudan University  and professor Xuanhe Wang at Dongbei University of Finance and Economics for insightful comments and/or discussions in an earlier version of the manuscript.

\section{Supplementary Material}
The technical proofs of the main theorems are put into the supplementary material.

\bibliographystyle{Chicago}
\bibliography{paper-ref}

\begin{thebibliography}{}

\bibitem[\protect\citeauthoryear{Aan{\ae}s, Fisker, {\AA}str\"{o}m, and
  Carstensen}{Aan{\ae}s et~al.}{2002}]{Aan2002robust}
Aan{\ae}s, H., R.~Fisker, K.~{\AA}str\"{o}m, and J.~M. Carstensen (2002).
\newblock Robust factorization.
\newblock {\em  IEEE Transactions on Pattern Analysis and Machine
  Intelligence\/}~{\em 24\/}(9), 1215--1225.

\bibitem[\protect\citeauthoryear{Ahn and Horenstein}{Ahn and
  Horenstein}{2013}]{Ahn2013Eigenvalue}
Ahn, S.~C. and A.~R. Horenstein (2013).
\newblock Eigenvalue ratio test for the number of factors.
\newblock {\em Econometrica\/}~{\em 81\/}(3), 1203--1227.

\bibitem[\protect\citeauthoryear{A\"{\i}t-Sahalia and Xiu}{A\"{\i}t-Sahalia and
  Xiu}{2017}]{Ait2017Using}
A\"{\i}t-Sahalia, Y. and D.~Xiu (2017).
\newblock Using principal component analysis to estimate a high dimensional
  factor model with high-frequency data.
\newblock {\em Journal of Econometrics\/}~{\em 201}, 384--399.

\bibitem[\protect\citeauthoryear{Ando and Bai}{Ando and
  Bai}{2020}]{Ando2020Quantile}
Ando, T. and J.~Bai (2020).
\newblock Quantile co-movement in financial markets: A panel quantile model
  with unobserved heterogeneity.
\newblock {\em Journal of the American Statistical Association\/}~{\em 115},
  266--279.

\bibitem[\protect\citeauthoryear{Bai}{Bai}{2003}]{Bai2003Inferential}
Bai, J. (2003).
\newblock Inferential theory for factor models of large dimensions.
\newblock {\em Econometrica\/}~{\em 71\/}(1), 135--171.

\bibitem[\protect\citeauthoryear{Bai and Ng}{Bai and
  Ng}{2002}]{Bai2002Determining}
Bai, J. and S.~Ng (2002).
\newblock Determining the number of factors in approximate factor models.
\newblock {\em Econometrica\/}~{\em 70\/}(1), 191--221.

\bibitem[\protect\citeauthoryear{Bickel and Levina}{Bickel and
  Levina}{2008}]{Beckel2008Cov}
Bickel, P.~J. and E.~Levina (2008).
\newblock Covariance regularization by thresholding.
\newblock {\em The Annals of Statistics\/}~{\em 36\/}(6), 2577--2604.

\bibitem[\protect\citeauthoryear{Calzolari and Halbleib}{Calzolari and
  Halbleib}{2018}]{Calzolari2018Estimating}
Calzolari, G. and R.~Halbleib (2018).
\newblock Estimating stable latent factor models by indirect inference.
\newblock {\em Journal of Econometrics\/}~{\em 205}, 280--301.

\bibitem[\protect\citeauthoryear{Chamberlain}{Chamberlain}{1983}]{Chamberlain1983}
Chamberlain, G. (1983).
\newblock A characterization of the distributions that imply mean-variance
  utility functions.
\newblock {\em Journal of Economic Theory\/}~{\em 29}, 975--988.

\bibitem[\protect\citeauthoryear{Chamberlain and Rothschild}{Chamberlain and
  Rothschild}{1983}]{Chamberlain1983Arbitrage}
Chamberlain, G. and M.~Rothschild (1983).
\newblock Arbitrage, factor structure, and mean-variance analysis on large
  asset markets.
\newblock {\em Econometrica\/}~{\em 51\/}(5), 1281--1304.

\bibitem[\protect\citeauthoryear{Chen, Mykland, and Zhang}{Chen
  et~al.}{2020}]{Chen2020The}
Chen, D., P.~Mykland, and L.~Zhang (2020).
\newblock The five trolls under the bridge: Principal component analysis with
  asynchronous and noisy high frequency data.
\newblock {\em Journal of the American Statistical Association, to appear\/}.

\bibitem[\protect\citeauthoryear{Chen, Dolado, and Gonzalo}{Chen
  et~al.}{2020}]{Chen2019Quantile}
Chen, L., J.~J. Dolado, and J.~Gonzalo (2020).
\newblock Quantile factor models.
\newblock {\em arXiv: 1911.02173\/}.

\bibitem[\protect\citeauthoryear{Fan, Liao, and Mincheva}{Fan
  et~al.}{2013}]{fan2013large}
Fan, J., Y.~Liao, and M.~Mincheva (2013).
\newblock Large covariance estimation by thresholding principal orthogonal
  complements.
\newblock {\em Journal of the Royal Statistical Society: Series B (Statistical
  Methodology)\/}~{\em 75\/}(4), 603--680.

\bibitem[\protect\citeauthoryear{Forni, Hallin, Lippi, and Reichlin}{Forni
  et~al.}{2000}]{forni2000generalized}
Forni, M., M.~Hallin, M.~Lippi, and L.~Reichlin (2000).
\newblock The generalized dynamic-factor model: Identification and estimation.
\newblock {\em Review of Economics and Statistics\/}~{\em 82\/}(4), 540--554.

\bibitem[\protect\citeauthoryear{Hallin and Liska}{Hallin and
  Liska}{2007}]{hallin2007determining}
Hallin, M. and R.~Liska (2007).
\newblock Determining the number of factors in the general dynamic factor
  model.
\newblock {\em Journal of the American Statistical Association\/}~{\em
  102\/}(478), 603--617.

\bibitem[\protect\citeauthoryear{He, Kong, Yu, and Zhang}{He
  et~al.}{2020}]{He2020Large}
He, Y., X.~Kong, L.~Yu, and X.~Zhang (2020).
\newblock Large-dimensional factor analysis without moment constraints,
  manuscript.
\newblock {\em Journal of Business and Economic Statistics, to appear\/}.

\bibitem[\protect\citeauthoryear{Ke and Kanade}{Ke and
  Kanade}{2005}]{Ke2005robust}
Ke, Q. and T.~Kanade (2005).
\newblock Robust $l_1$ norm factorization in the presence of outliers and
  missing data by alternative convex programming.
\newblock {\em Proceedings of the 2005 IEEE computer society conference on
  computer vision and pattern recognition\/}~{\em CVPR'05}, 1063--1069.

\bibitem[\protect\citeauthoryear{Kelly, Pruitt, and Su}{Kelly
  et~al.}{2019}]{Kelly2019Characteristics}
Kelly, B.~T., S.~Pruitt, and Y.~Su (2019).
\newblock Characteristics are covariances: A unified model of risk and return.
\newblock {\em Journal of Financial Economics\/}~{\em 134}, 501--524.

\bibitem[\protect\citeauthoryear{Kim and Wang}{Kim and
  Wang}{2016}]{Kim2016Sparse}
Kim, D. and Y.~Wang (2016).
\newblock Sparse pca based on high-dimensional ito processes with measurement
  errors.
\newblock {\em Journal of Multivariate analysis\/}~{\em 152}, 172--189.

\bibitem[\protect\citeauthoryear{Knight}{Knight}{1998}]{Knight1998Asymptotics}
Knight, K. (1998).
\newblock Asymptotics for $l_1$ regression estimators under general conditions.
\newblock {\em Annals of Statistics\/}~{\em 26\/}(2), 755--770.

\bibitem[\protect\citeauthoryear{Kong}{Kong}{2017}]{Kong2017On}
Kong, X.~B. (2017).
\newblock On the number of common factors with high-frequency data.
\newblock {\em Biometrika\/}~{\em 104\/}(2), 397--410.

\bibitem[\protect\citeauthoryear{Kong}{Kong}{2018}]{Kong2018On}
Kong, X.~B. (2018).
\newblock On the integrated idiosyncratic and systematic volatility with the
  large panel high-frequency data.
\newblock {\em Annals of Statistics\/}~{\em 46}, 1077--1108.

\bibitem[\protect\citeauthoryear{Kong and Liu}{Kong and
  Liu}{2018}]{Kong2018Testing}
Kong, X.~B. and C.~Liu (2018).
\newblock Testing against constant factor loading matrix with large panel
  high-frequency data.
\newblock {\em Journal of Econometrics\/}~{\em 204\/}(2), 301--319.

\bibitem[\protect\citeauthoryear{Kong, Liu, and Jing}{Kong
  et~al.}{2015}]{kong2015testing}
Kong, X.~B., Z.~Liu, and B.-Y. Jing (2015).
\newblock Testing for pure-jump processes for high-frequency data.
\newblock {\em The Annals of Statistics\/}~{\em 43\/}(2), 847--877.

\bibitem[\protect\citeauthoryear{Kong, Wang, Xing, Xu, and Ying}{Kong
  et~al.}{2019}]{Kong2019Factor}
Kong, X.~B., J.~Wang, J.~Xing, C.~Xu, and C.~Ying (2019).
\newblock Factor and idiosyncratic empirical processes.
\newblock {\em Journal of the American Statistical Association\/}~{\em
  114\/}(527), 1138--1146.

\bibitem[\protect\citeauthoryear{Onatski}{Onatski}{2009}]{onatski2009testing}
Onatski, A. (2009).
\newblock Testing hypotheses about the number of factors in large factor
  models.
\newblock {\em Econometrica\/}~{\em 77\/}(5), 1447--1479.

\bibitem[\protect\citeauthoryear{Owen and Rabinovitch}{Owen and
  Rabinovitch}{2012}]{Owen2012on}
Owen, J. and R.~Rabinovitch (2012).
\newblock On the class of elliptical distributions and their applications to
  the theory of portfolio choice.
\newblock {\em Journal of Finance\/}~{\em 38}, 745--752.

\bibitem[\protect\citeauthoryear{Pelger}{Pelger}{2018}]{Pelger2018Large}
Pelger, M. (2018).
\newblock Large-dimensional factor modeling based on high-frequency
  observations.
\newblock {\em Journal of Econometrics\/}~{\em 208}, 23--42.

\bibitem[\protect\citeauthoryear{Pollard}{Pollard}{1991}]{Pollard1991Asymptotics}
Pollard, D. (1991).
\newblock Asymptotics for least absolute deviation regression estimators.
\newblock {\em Econometric Theory\/}~{\em 7\/}(2), 186--199.

\bibitem[\protect\citeauthoryear{Stock and Watson}{Stock and
  Watson}{2002a}]{stock2002forecast}
Stock, J.~H. and M.~W. Watson (2002a).
\newblock Forecasting using principal components from a large number of
  predictors.
\newblock {\em Journal of the American Statistical Association\/}~{\em 97},
  1167--1179.

\bibitem[\protect\citeauthoryear{Stock and Watson}{Stock and
  Watson}{2002b}]{Stock2002Macroeconomic}
Stock, J.~H. and M.~W. Watson (2002b).
\newblock Macroeconomic forecasting using diffusion indexes.
\newblock {\em Journal of Business \& Economic Statistics\/}~{\em 20\/}(2),
  147--162.

\bibitem[\protect\citeauthoryear{Tankov and Cont}{Tankov and
  Cont}{2004}]{Tancov2004Financial}
Tankov, P. and R.~Cont (2004).
\newblock Financial modelling with jump processes, \text{CRC} \text{P}ress,
  \text{B}oca \text{R}aton.

\bibitem[\protect\citeauthoryear{Trapani}{Trapani}{2018}]{Trapani2018A}
Trapani, L. (2018).
\newblock A randomised sequential procedure to determine the number of factors.
\newblock {\em Journal of the American Statistical Association\/}~{\em 113},
  1341--1349.

\bibitem[\protect\citeauthoryear{Yu, He, and Zhang}{Yu
  et~al.}{2019}]{yu2019robust}
Yu, L., Y.~He, and X.~Zhang (2019).
\newblock Robust factor number specification for large-dimensional elliptical
  factor model.
\newblock {\em Journal of Multivariate analysis\/}~{\em 174}, 104543.

\end{thebibliography}
\clearpage

	\section*{Supplementary Material for ``Factor Analysis without Moment Constraint" }

The supplementary material contains all the technical proofs for the main theorems in ``Robust Factor Analysis without Moment Constraint". In Section \ref{sec:A}, we provide some useful lemmas and corollaries and the detailed proofs of main theorems are given in Section \ref{sec:B}.

\begin{appendices}
	\section{Useful Lemmas and Corollaries}\label{sec:A}

Let $C$ be a generic constant that will be used in deriving upper bounds, and it may take different values in different places. $E_f$ stands for the expectation conditional on $\bbf_t^0$'s. Define
$$
\bW_0=\Big\{\sum^p_{i=1}h_i(0)\bl_i\bl_i^{\prime}\Big\}^{-1}\sum^p_{i=1}h_i(0)\bl_i\bl_i^{0\prime}.
$$
Assumption 1 and Assumption 3 imply that $\|\bW_0\|\leq C$. Reparameterize $\bl_i$'s and $\bbf_t$'s with
$u_i=\bW_0^{\prime}\bl_i-\bl_i^0$ and $v_t=\bbf_t-\bW_0\bbf_t^0$. One easily deduces the decomposition as follows,
$$
\tau_{it}=:\bl_i^{\prime}\bbf_t-\bl_i^{0\prime}\bbf_t^0=\bl_i^{\prime}v_t+u_i^{\prime}\bbf_t^0.
$$
Notice here that $\bl_i$ is still related to $u_i$ and $u_i$ simply serves as a measure of distance from $\bl_i$ to $\bl_i^0$. By the mathematical expression of the static approximate factor model,
\[
(\bY_t)_{p\times 1}=\bL_{p\times r}(\bbf_t)_{r\times 1}+(\bepsilon_t)_{p\times 1}, \ t=1,\ldots, T,
\]
we have that
\begin{eqnarray*}
\sum^p_{i=1}\sum^T_{t=1}\rho_{\tau}(y_{it}-\bl_i^{\prime}\bbf_t)=\sum^p_{i=1}\sum^T_{t=1}\rho_{\tau}\Big(\epsilon_{it}-\big(\bl_i^{\prime}\bbf_t-\bl_i^{0\prime}\bbf_t^0\big)\Big).
\end{eqnarray*}
Notice that, by Assumption {\bf 1}',
\begin{eqnarray}\label{equivalence1}
(\hat{\bL}, \hat{\bF})&=&\bargmin_{\bL,\bF}\sum^p_{i=1}\sum^T_{t=1}\rho_{\tau}(\epsilon_{it}-\tau_{it})\nonumber\\
&=&\bargmin_{\bL,\bF}\sum^p_{i=1}\sum^T_{t=1}\Big(\rho_{\tau}(\epsilon_{it}-(\bl_i^{\prime}\bbf_t-\bl_i^{0\prime}\bbf_t^0))-\rho_{\tau}(\epsilon_{it})\Big)\nonumber\\
&=&\bargmin_{\bL, \bF}\sum^p_{i=1}\sum^T_{t=1}\Big(\rho_{\tau}\big(\epsilon_{it}-(\bl_i^{\prime}v_t-u_i^{\prime}\bbf_t^0)\big)-\rho_{\tau}(\epsilon_{it})\Big)\nonumber\\
&=:&\bargmin_{\bL,\bF}\sum^p_{i=1}\sum^T_{t=1}g_l(u_i, v_t)\nonumber\\
&=&\bargmin_{\bL, \bF}\sum^p_{i=1}\sum^T_{t=1}\Big(\rho_{\tau}\big(\epsilon_{it}-(u_i^{\prime}\bW_0^{-1}\bbf_t+\bl_i^{0\prime}\bW_0^{-1}v_t)\big)-\rho_{\tau}(\epsilon_{it})\Big)\nonumber\\
&=:&\bargmin_{\bL, \bF}\sum^p_{i=1}\sum^T_{t=1}g_f(u_i, v_t),
\end{eqnarray}
where $\bL=(\bl_1,\ldots,\bl_p)^{\prime}$, $\bF=(\bbf_1,\ldots,\bbf_T)$.

One easily deduces the equality that
\begin{equation}\label{equality}
\rho_{\tau}(x-u)-\rho_{\tau}(x)=u\big(I(x\leq 0)-\tau\big)+(u-x)\big(I(x\leq u)-I(x\leq 0)\big).
\end{equation}
This shows that
\begin{equation}\label{error}
g_l(u_i,v_t)=:\rho_{\tau}(\epsilon_{it}-\tau_{it})-\rho_{\tau}(\epsilon_{it})=\tau_{it}D_{it}+R_{it},
\end{equation}
where $D_{it}=I(\epsilon_{it}\leq 0)-\tau$ and
$
R_{it}=(\tau_{it}-\epsilon_{it})\{I(\epsilon_{it}\leq \tau_{it})-I(\epsilon_{it}\leq 0)\}.
$
$g_l(u_i, v_t)$ can also be expressed as
\begin{equation}\label{error}
g_l(u_i,v_t)=\tilde{\tau}_{it}\tilde{D}_{it}+\tilde{R}_{it}+\rho_{\tau}(\epsilon_{it}-u_i^{\prime}\bbf_t^0)-\rho_{\tau}(\epsilon_{it}),
\end{equation}
where $\tilde{D}_{it}=I(\tilde{\epsilon}_{it}\leq 0)-\tau$ with $\tilde{\epsilon}_{it}=\epsilon_{it}-u_i^{\prime}\bbf_t^0$, $\tilde{\tau}_{it}=\bl_i^{\prime}v_t=:p^{-1/2}\bl_i^{\prime}\tilde{v}_t$ and
$
\tilde{R}_{it}=(\tilde{\tau}_{it}-\tilde{\epsilon}_{it})\{I(\tilde{\epsilon}_{it}\leq \tilde{\tau}_{it})-I(\tilde{\epsilon}_{it}\leq 0)\}
$.
There are some facts on $\tilde{R}_{it}$. It is nonnegative, piecewise linear, monotone in $\tilde{\tau}_{it}$, and bounded when $|\tilde{\tau}_{it}|\leq C$ due to Assumption {\bf 1}'.

Since $\tau\in (0, 1)$ is fixed, without loss of generality, we set $\tau=1/2$ in the sequel of the proof. Now we are going to prove that for arbitrarily fixed $\bl^{(0)}_i$'s (initial guess of $\bl_i$'s in the alternating iterative algorithm) lying in the parameter space under Assumption {\bf 1}' and for given $\bbf_t^0$, $\hat{v}^{(1)}_t$ satisfies $\|\hat{v}^{(1)}_t\|=:\|\hat{\bbf}^{(1)}_t-\bW^{(0)}\bbf_t^0\|\leq Cp^{-1/2}$ for some $r\times r$ matrix $\bW^{(0)}$ dependent only on $\bl_i^{(0)}$ and $\bbf_t^0$, where $\hat{\bbf}_t^{(1)}$ is the optimal solution of $\bbf_t$ to minimizing $\sum_{i=1}^pg_l(u^{(0)}_i, v_t)$ where $u_i^{(0)}=:\bW^{(0)\prime}\bl_i^{(0)}-\bl_i^0$ is fixed. Let $\tilde{\tau}_{it}^{(0)}$ and $\tilde{R}_{it}^{(0)}$ be similarly defined as $\tilde{\tau}_{it}$ and $\tilde{R}_{it}$ except for replacing $\bl_i$ and $\bW_0$ by $\bl_i^{(0)}$ and $\bW^{(0)}$ defined below (\ref{eI}), respectively. That being said, our first lemma gives the theoretical property of the optimal solution to the cross-sectional regression in the distance from $\bbf_t$ to $\bbf_t^0$, for initially given design matrix $\bL^{(0)}=(\bl_1^{(0)},\ldots,\bl_p^{(0)})^{\prime}$.

\begin{lemma}\label{lemma1}
For fixed $\bl^{(0)}_i$'s and given $\bbf_t^0$, under Assumptions 1-3,
$$
P\left(\|\sqrt{p}\hat{v}^{(1)}_t-\overline{v}^{(1)}_t\|> \delta\right)\rightarrow 0,
$$
for any $\delta>0$, where
$$
\overline{v}^{(1)}_t=-\frac{1}{2}\Big\{\frac{1}{p}\sum^p_{i=1}h_i\big(u_i^{(0)\prime}\bbf_t^0\big)\bl^{(0)}_i\bl_i^{(0)\prime}\Big\}^{-1}\sum^p_{i=1}\big(\tilde{D}^{(0)}_{it}-E_f\tilde{D}^{(0)}_{it}\big)\frac{\bl_i^{(0)}}{\sqrt{p}},
$$
where $\tilde{D}_{it}^{(0)}$ is similarly defined as $\tilde{D}_{it}$ except for replacing $u_i$ by its initial value  $u_i^{(0)}$.

Moreover, if Assumptions 1-5 are satisfied, then we further have
$$
P\left(\max_{t\leq T}\|\sqrt{p}\hat{v}^{(1)}_t-\overline{v}^{(1)}_t\|> \delta\right)\rightarrow 0.
$$
\end{lemma}

\begin{proof}
First we give an expansion of $\sum^p_{i=1}g_l(u_i^{(0)}, v_t)$ in $\{v_t; \|\sqrt{p}v_t\|\leq M\}$ for fixed $\bl_i^{(0)}$'s and given $\bbf_t^0$ satisfying Assumption 1. (\ref{error}) shows that
\begin{eqnarray}\label{exp}
\sum^p_{i=1}g_l(u_i^{(0)}, v_t)&=&\sum^p_{i=1}\tilde{D}^{(0)}_{it}\bl_i^{(0)\prime}p^{-1/2}\tilde{v}_t+\sum^{p}_{i=1}E_f\tilde{R}^{(0)}_{it}+\sum^{p}_{i=1}\big(\tilde{R}^{(0)}_{it}-E_f\tilde{R}^{(0)}_{it}\big)\nonumber\\
&&+\sum^p_{i=1}\big(|\epsilon_{it}-u^{\prime}_i\bbf_t^0|-|\epsilon_{it}|\big)/2,
\end{eqnarray}
where $\tilde{v}_t=\sqrt{p}v_t$ and $\|\tilde{v}_t\|\leq M$. Because $\sum^{p}_{i=1}(|\epsilon_{it}-u_i^{\prime}\bbf_t^0|-|\epsilon_{it}|)/2$ is irrelevant to optimization in $v_t$, we ignore this term below. Now we analyze (\ref{exp}) term by term. For the first term of (\ref{exp}),
\begin{eqnarray}\label{eI}
&&\sum^p_{i=1}E_f\tilde{D}^{(0)}_{it}\bl_i^{(0)\prime}p^{-1/2}\tilde{v}_t=\sum^p_{i=1}\bl_i^{(0)\prime}p^{-1/2}\tilde{v}_t\left(P_f(\epsilon_{it}< u_i^{(0)\prime}\bbf_t^0)-1/2\right)\nonumber\\
&=&\sum^p_{i=1}\bl_i^{(0)\prime}p^{-1/2}\tilde{v}_tP_f\left(0<\epsilon_{it}< u_i^{(0)\prime}\bbf_t^0\right)=\sum^p_{i=1}\bl_i^{(0)\prime}p^{-1/2}\tilde{v}_t\int^{u_i^{(0)\prime}\bbf_t^0}_0h_i(x)dx\nonumber\\
&=& p^{-1/2}\tilde{v}_t^{\prime}\sum^p_{i=1}h_i(\xi^{(0)}_{it})\bl_i^{(0)}(\bl_i^{(0)\prime}\bW^{(0)}-\bl_i^{0\prime})\bbf_t^0=0,
\end{eqnarray}
where $\xi_{it}^{(0)}$ is some variable in $(0, u_i^{(0)\prime}\bbf_t^0)$ and in the last equality we have used the definition that
$$
\bW^{(0)}=:\Big\{\sum^p_{i=1}h_i(\xi_{it}^{(0)})\bl_i^{(0)}\bl_i^{(0)\prime}\Big\}^{-1}\sum^{p}_{i=1}h_i(\xi_{it}^{(0)})\bl_i^{(0)}\bl_i^{0\prime}.
$$
Notice that $\bW^{(0)}$ depends on $t$, but for simplicity of notation and easy comparing with $\bW_0$, we suppress the subscript $t$ and simply write $\bW^{(0)}_t=\bW^{(0)}$. Assumption {\bf 1}' and Assumption 3 and the restriction $\{\|\tilde{v}_t\|\leq M\}$ guarantee that $P(\|\bW^{(0)}\|\leq C)\rightarrow 1$ as $p, T\rightarrow \infty$. In the sequel, we restrict that $\|\bW^{(0)}\|\leq C$.

For the second term of (\ref{exp}),
\begin{eqnarray}\label{eR}
&&\sum^p_{i=1}\Big\{E_f\tilde{R}^{(0)}_{it}-h_i\big((\bl_i^{(0)\prime}\bW^{(0)}-\bl_i^{0\prime})\bbf_t^0\big)\big(\bl_i^{(0)\prime}p^{-1/2}\tilde{v}_t\big)^2\Big\}\nonumber\\
&\leq& C\sum^p_{i=1}\int_0^{\bl_i^{(0)\prime}\tilde{v}_tp^{-1/2}}\big(\bl_i^{(0)\prime}\tilde{v}_tp^{-1/2}-x\big)xdx\leq C\sum^p_{i=1}\big(\bl_i^{(0)\prime}\tilde{v}_tp^{-1/2}\big)^3=o_p(1),
\end{eqnarray}
by Assumption {\bf 1}'. Here $o_p(1)$ holds uniformly in $t\leq T$ for $\|\tilde{v}_t\|\leq M$. For the third term of (\ref{exp}), we are going to prove that \begin{equation}\label{R-eR}
\left|\sum^p_{i=1}(\tilde{R}_{it}^{(0)}-E_f\tilde{R}^{(0)}_{it})\right|=o_p(1),
\end{equation}
where $o_p(1)$ holds uniformly in $\|\tilde{v}_t\|\leq M$ for fixed $\{\bl_i^{(0)}\}$'s. Since $v_t$ is of fixed dimension, with out of loss of generality and for simplicity of notation we assume here $r=1$ in proving (\ref{R-eR}). To this end, we split the range of $\tilde{v}_t$, $(-M, M]$, into non-overlapping intervals $(C_{k}, C_{k+1}]$ so that $C_{k+1}-C_k=\delta^{\prime}/M$. Then the number of subintervals is $2M^2/\delta^{\prime}$. For convenience, we rewrite $\tilde{R}_{it}^{(0)}$ as $R(u_i^{(0)}, \tilde{v}_t, \epsilon_{it})$. Then
\begin{eqnarray}\label{supR}
&&\sup_{\{\|\tilde{v}_t\|\leq  C\}}\left|\sum^p_{i=1}(\tilde{R}_{it}^{(0)}-E_f\tilde{R}_{it}^{(0)})\right|\nonumber\\
&=& \max_{\{k\}}\sup_{\{\tilde{v}_t\in (C_{k}, C_{k+1}]\}}\left|\sum^p_{i=1}\Big\{R(u_i^{(0)}, \tilde{v}_t, \epsilon_{it})-E_fR(u_i^{(0)}, \tilde{v}_t, \epsilon_{it})\Big\}\right|.
\end{eqnarray}
Notice that $R(u, \tilde{v}_t, \epsilon)$ is monotone in $\tilde{v}_t$ when $u$ is fixed,
$$
0<R(u, C_k, \epsilon)\wedge R(u, C_{k+1}, \epsilon)\leq \max_{\tilde{v}_t\in (C_k,C_{k+1}]}R(u, \tilde{v}_t, \epsilon)\leq R(u, C_{k}, \epsilon)\vee R(u, C_{k+1}, \epsilon),
$$
where the results of the operators $\wedge$ and $\vee$ depend only on $sign(\tilde{\epsilon}_{it})$.
Let $$
G_{bk}(u,\epsilon)=R(u, C_k, \epsilon)\vee R(u, C_{k+1}, \epsilon)$$
and
$$
 G_{sk}(u, \epsilon)=R(u, C_k, \epsilon)\wedge R(u, C_{k+1}, \epsilon),
$$
which are bivariate functions bounded by $C$ for $\|\overline{v}_t\|\leq M$. $G_{bk}(u_i^{(0)},\epsilon_{it})$ and $G_{sk}(u_i^{(0)}, \epsilon_{it})$ are simply the values of $G_{bk}(u, \epsilon)$ and $G_{sk}(u, \epsilon)$ realized at $(u,\epsilon)=(u_i^{(0)},\epsilon_{it})$. Then the right hand side of (\ref{supR}) is less than
\begin{eqnarray}\label{II}
&&\max_{k}\left|\sum^p_{i=1}\Big(G_{bk}(u_i^{(0)},\epsilon_{it})-E_fG_{sk}(u_i^{(0)},\epsilon_{it})\Big)\right|+\max_k\left|\sum^p_{i=1}\Big(G_{sk}(u_i^{(0)},\epsilon_{it})-E_fG_{bk}(u_i^{(0)},\epsilon_{it})\Big)\right|\nonumber\\
&\leq &\max_{k}\left|\sum^p_{i=1}\Big(G_{bk}(u_i^{(0)},\epsilon_{it})-E_fG_{bk}(u_i^{(0)},\epsilon_{it})\Big)\right|
+\max_{k}\left|\sum^p_{i=1}\Big(E_fG_{bk}(u_i^{(0)},\epsilon_{it})-E_fG_{sk}(u_i^{(0)},\epsilon_{it})\Big)\right|\nonumber\\
&&+\max_k\left|\sum^p_{i=1}\Big(E_fG_{bk}(u_i^{(0)},\epsilon_{it})-E_fG_{sk}(u_i^{(0)},\epsilon_{it})\Big)\right|
+\max_{k}\left|\sum^p_{i=1}\Big(E_fG_{sk}(u_i^{(0)},\epsilon_{it})-G_{sk}(u_i^{(0)},\epsilon_{it})\Big)\right|\nonumber\\
&=:&I_{1t}+I_{2t}+II_{1t}+II_{2t}.
\end{eqnarray}
Assumption {\bf 1}' and (\ref{eR}) show that
\begin{eqnarray}\label{Imean}
I_{2t}+II_{1t}&\leq & C\delta^{\prime}\left\|p^{-1}\sum^p_{i=1}h_i(u_i^{(0)\prime}\bbf_t^0)\bl_i^{(0)}\bl_i^{(0)\prime}\right\|\nonumber\\
&=&C\delta^{\prime}\left\|p^{-1}\sum^p_{i=1}h_i\left(\bl_i^{(0)\prime}\bbf_t-\bl_i^{0\prime}\bbf_t^0+O\Big(\frac{M}{\sqrt{p}}\Big)\right)\bl_i^{(0)}\bl_i^{(0)\prime}\right\|.\label{I}
\end{eqnarray}
where $O({M}/{\sqrt{p}})$ holds uniformly in $i\leq p$. Assumptions {\bf 1}' and 3 yield
\begin{eqnarray}\label{Ivar}
I_{1t}+II_{2t}&=&O_p\Big(\frac{2M^2}{\delta^{\prime}}\Big)^r\max_i\Big\|\frac{\bl_i^{(0)}}{\sqrt{p}}\Big\|^{\frac{1}{2}}E^{\frac{1}{2}}\left\|\sum^p_{i=1}\frac{h_i(u_i^{(0)\prime}\bbf_t^0)\bl_i^{(0)}\bl_i^{(0)\prime}}{p}\right\|=o_p(1).\label{II}
\end{eqnarray}
This proves (\ref{R-eR}) by letting $p\rightarrow \infty$ first and then $\delta^{\prime}\rightarrow 0$. Summarizing the results for all three terms of (\ref{exp}), we have, by ignoring $\sum^p_{i=1}\big(|\epsilon_{it}-u_i^{\prime}\bbf_t^0|-|\epsilon_{it}|\big)/2$,
\begin{eqnarray}\label{exp1}
&&\sum^p_{i=1}g_l(u_i^{(0)}, v_t)=: \sum^p_{i=1}G_l(u_i^{(0)}, \tilde{v}_t)\nonumber\\
&=&\sum^p_{i=1}\big(\tilde{D}^{(0)}_{it}-E_f\tilde{D}^{(0)}_{it}\big)\frac{\bl_i^{(0)\prime}}{\sqrt{p}}\tilde{v}_t+\sum^{p}_{i=1}h_i\Big(\big(\bW^{(0)\prime}\bl_i^{(0)}-\bl_i^{0}\big)^{\prime}\bbf_t^0\Big)\Big(\frac{\bl_i^{(0)\prime}\tilde{v}_t}{\sqrt{p}}\Big)^2+o_p(1)\nonumber\\
&=& (\tilde{v}_t-\overline{v}^{(1)}_t)^{\prime}\sum^p_{i=1}h_i(u_i^{(0)\prime}\bbf_t^0)\frac{\bl_i^{(0)}\bl_i^{(0)\prime}}{p}(\tilde{v}_t-\overline{v}^{(1)}_t)-\frac{1}{4}\Big\{\sum^p_{i=1}(\tilde{D}^{(0)}_{it}-E_f\tilde{D}^{(0)}_{it})\frac{\bl_i^{(0)}}{\sqrt{p}}\Big\}^{\prime}\nonumber\\
&&\times\Big\{\sum^p_{i=1}h_i(u_i^{(0)\prime}\bbf_t^0)\frac{\bl_i^{(0)}\bl_i^{(0)\prime}}{p}\Big\}^{-1}\Big\{\sum^p_{i=1}(\tilde{D}^{(0)}_{it}-E_f\tilde{D}^{(0)}_{it})\frac{\bl_i^{(0)}}{\sqrt{p}}\Big\}+o_p(1),
\end{eqnarray}
where $o_p(1)$ holds uniformly in $\|\tilde{v}_t\|\leq M$. (\ref{exp1}) demonstrates that $\tilde{v}_t$ achieves the minimum $\overline{v}^{(1)}_t$ asymptotically whenever $\|\tilde{v}_t\|\leq M$. Let $\sqrt{p}\hat{v}^{(1)}_t$ be the minimizer of $\sum^p_{i=1}G_l(u_i^{(0)}, \tilde{v}_t)$ over $\{\tilde{v}_t\in R^r\}$ for fixed $u_i^{(0)}$ and $\bl_i^{(0)}$ and given $\bbf_t^0$. Let $\breve{v}_t=\overline{v}^{(1)}_t+\beta_te_t$ where $e_t$ is a vector of unit length, and let $v_t^*=\overline{v}^{(1)}_t+\delta e_t$ so that $v_t^*$ lies in the line segment from $\overline{v}^{(1)}_t$ to $\breve{v}_t$. Notice that $\sum^p_{i=1}G_l(u_i^{(0)}, \tilde{v}_t)$ is a convex function in $\tilde{v}_t$ given $\bl_i^{(0)}$'s and $\bbf_t^0$. For $\beta_t>\delta$, by convexity, (\ref{eR}), (\ref{R-eR}) and (\ref{exp1}), and restricted on  $\{\|\overline{v}^{(1)}_t\|\leq M-\delta\}$,
\begin{eqnarray}\label{convex1}
&&\sum^{p}_{i=1}\Big[G_l(u_i^{(0)}, \breve{v}_t)-G_l(u_i^{(0)}, \overline{v}^{(1)}_t)\Big]\nonumber\\
&> & \frac{\beta_t}{\delta}\sum^{p}_{i=1}\Big[G_l(u_i^{(0)}, v^*_t)-G_l(u_i^{(0)}, \overline{v}^{(1)}_t)\Big]\nonumber\\
&> &\delta \beta_t e_t^{\prime}\Big\{\frac{1}{p}\sum^p_{i=1}h_i(u_i^{(0)\prime}\bbf_t^0)\bl^{(0)}_i\bl_i^{(0)\prime}\Big\}e_t-\delta^{-1}|o_p(1)|,
\end{eqnarray}
as $p, T\rightarrow \infty$, where $o_p(1)$ holds uniformly in $\{\|\overline{v}^{(1)}_t\|\leq M-\delta\}$. Then by Assumption {\bf 1}' and Assumption 3(2) and (\ref{convex1}), for any $\delta>0$,
\begin{eqnarray}\label{Chebyshev}
P\Big\{\|\sqrt{p}\hat{v}^{(1)}_t-\overline{v}^{(1)}_t\|>\delta\Big\}&\leq & P\Big\{\sum^p_{i=1}\big[G_l(u_i^{(0)},\breve{v}_t)-G_l(u_i^{(0)},\overline{v}^{(1)}_t)\big]<0, \|\overline{v}^{(1)}_t\|\leq M-\delta\Big\} \nonumber\\
&& + P\Big\{\|\overline{v}^{(1)}_t\|> M-\delta\Big\} \leq \epsilon,
\end{eqnarray}
for arbitrarily small $\epsilon>0$, where we have used the Chebyshev inequality
\begin{equation}\label{cheby1}
P\Big\{\|\overline{v}^{(1)}_t\|> M-\delta\Big\}\leq (M-\delta)^{-2}E(\overline{v}^{(1)}_t)^2\leq \epsilon/2,
\end{equation}
by choosing $M$ large enough and Assumption 3(2).

Next, we prove the uniform result in $t\leq T$. Taking $M=C\log{T}$ for $C$ large enough, by the Markov inequality,
\begin{eqnarray}
P\Big\{\max_{t\leq T}\|\overline{v}^{(1)}_t\|>C\log{T}\Big\}\leq Te^{-C\log{T}}\max_tE\Big\{E_f\exp\{\overline{v}^{(1)}_t\}\Big\}=o(1),
\end{eqnarray}
due to Assumption 4. Hence, in the sequel, we restrict on the set $\{\max_t\|\overline{v}^{(1)}_t\|\leq C\log{T}\}$. Repeating the steps of for the non-uniform results, we find that (\ref{I}) still holds uniformly in $t\leq T$, i.e.,
\begin{equation}\label{uI}
\max_{t\leq T} (I_{2t}+II_{1t})\leq C\delta^{\prime}.
\end{equation}
Parallel to (\ref{II}), the Markov inequality and Assumptions 4 and 5 show that
\begin{eqnarray}\label{uII}
P\Big\{\max_{t\leq T} (I_{1t}+II_{2t})>\epsilon\Big\}\leq CT(\log{T})^{2r}e^{-\epsilon p^{1/4}}=o(1).
\end{eqnarray}
This proves that under the more stringent condition on $p$ and $T$, (\ref{R-eR}) holds uniformly in $t\leq T$, and hence the $o_p(1)$ and $\epsilon$ terms in (\ref{exp1}) and (\ref{Chebyshev}) hold uniformly in $t\leq T$. Then paralleling to (\ref{Chebyshev}) proves the uniform (in $t$) results.

\end{proof}

Now, we alternate to fix $\hat{\bbf}_t^{(1)}$, $\bW^{(0)}$, $\bbf_t^0$, and thus $\hat{v}_t^{(1)}$, i.e., the optimal solution to the cross-sectional regression in $\tilde{v}_t$ done in Lemma \ref{lemma1}, and run time series regression in $u_i^{(1)}=:\bW^{(0)\prime}\bl_i-\bl_i^0$. We write $g_f(u_i^{(1)}, v_t^{(1)})=:\rho_{\tau}\Big(\epsilon_{it}-\big(u_i^{(1)\prime}(\bW^{(0)})^{-1}\bbf_t^{(1)}+\bl_i^{0\prime}(\bW^{(0)})^{-1}v^{(1)}_t\big)\Big)-\rho_{\tau}(\epsilon_{it})$ and set  $\tau^{(1)}_{it}=u_i^{(1)\prime}(\bW^{(0)})^{-1}\bbf^{(1)}_t+\bl_i^{0\prime}(\bW^{(0)})^{-1}v^{(1)}_t$. Let $R^{(1)}_{it}$ be similarly defined as $R_{it}$ except for replacing $\tau_{it}$ by $\tau^{(1)}_{it}$. Let $\hat{u}_i^{(1)}$ (and correspondingly $\hat{\bl}_i^{(1)}$) be the optimal solution to minimizing $\sum^T_{t=1}g_f(u_i^{(1)}, \hat{v}_t^{(1)})$ in $u_i^{(1)}$.

\begin{lemma}\label{lemma2}
Given $\bbf_t^0$, $\bW^{(0)}$ and $\hat{v}_t^{(1)}$'s, under Assumptions 1-5,
$$
\max_i\|\hat{u}^{(1)}_i\|=O_p\Big(\frac{\log{p}}{\sqrt{T}}\Big)+o_p\Big(\frac{1}{\sqrt{p}}\Big).
$$
Moreover, if further $\frac{T}{p}[(\log{p})^2(\log{T})^2+\frac{(\log{T})^3}{\sqrt{p}}]+\frac{(\log{p})^5}{\sqrt{T}}=o(1)$,
$$
P\left(\max_{i\leq p}\|\sqrt{T}\hat{u}^{(1)}_i-\overline{u}^{(1)}_i\|>\delta\right)\rightarrow 0,
$$
for any constant $\delta>0$, where
$$
\overline{u}^{(1)}_i=-\frac{1}{2h_i(0)}\Big(\sum^T_{t=1}\frac{\bbf_t^0\bbf_t^{0\prime}}{T}\Big)^{-1}\sum^T_{t=1}D_{it}\frac{\bbf_t^0}{\sqrt{T}}.
$$
\end{lemma}
\begin{proof}
Now, $g_f(u_i^{(1)}, v^{(1)}_t)$ can be rewritten as
$$
g_f(u_i^{(1)},v^{(1)}_t)=D_{it}\tau_{it}^{(1)}+E_fR^{(1)}_{it}+R_{it}^{(1)}-E_fR_{it}^{(1)}=:G_f(\tilde{u}_i^{(1)}, v^{(1)}_t),
$$
where $\tilde{u}_i^{(1)}=\sqrt{T}u_i^{(1)}$. Parallel to the proof of Lemma \ref{lemma1} and restricted on $\{\max_i\|\tilde{u}_i^{(1)}\|\leq M\}$,
\begin{equation}\label{dec2}
\sum^T_{t=1}G_f(\tilde{u}_i^{(1)}, v^{(1)}_t)=\sum^T_{t=1}D_{it}\tau^{(1)}_{it}+\sum^T_{t=1}E_fR_{it}^{(1)}+\sum^T_{t=1}(R_{it}^{(1)}-E_fR^{(1)}_{it}).
\end{equation}
For the second term in the right hand side of (\ref{dec2}), similar to (\ref{eR}) and by Lemma \ref{lemma1}, Assumption {\bf 1}' and Assumption 4,
\begin{eqnarray}\label{term2}
&&E\Big\{\sup_{\max_{t\leq T}\|v_t^{(1)}\|\leq C\log{T}/\sqrt{p}}\max_i
\Big|\sum^T_{t=1}E_fR^{(1)}_{it}-\sum^T_{t=1}h_i(0)\Big(\frac{\tilde{u}_i^{(1)\prime}(\bW^{(0)})^{-1}\bbf^{(1)}_t}{\sqrt{T}}+\bl_i^{0\prime}(\bW^{(0)})^{-1}v^{(1)}_t\Big)^2\Big|\Big\}\nonumber\\
&\leq & C  E\Big\{\sup_{\max_{t\leq T}\|v_t^{(1)}\|\leq C\log{T}/\sqrt{p}}\sum^T_{t=1}\max_i \Big|\frac{\tilde{u}_i^{(1)\prime}\bbf_t^{(1)}}{\sqrt{T}}+\bl_i^{0\prime}(\bW^{(0)})^{-1}v_t^{(1)}\Big|^3\Big\} \nonumber\\
&\leq &  C\Big(\frac{M^3}{\sqrt{T}}+\frac{T(\log{T})^3}{p^{3/2}}\Big).
\end{eqnarray}
For the third term of (\ref{dec2}), as in the proof of Lemma \ref{lemma1}, we assume $r=1$ and split the range of $\tilde{u}_i^{(1)}$, $(-M, M]$, into $2M^2/\delta^{\prime}$ non-overlapping subintervals $(C_k, C_{k+1}]$ with $C_{k+1}-C_{k}=\delta^{\prime}/M$. Rewrite $R^{(1)}_{it}=R(\tilde{u}_i^{(1)}, v^{(1)}_t, \epsilon_{it})$. By the monotonicity of $R(\tilde{u}_i^{(1)}, v_t^{(1)}, \epsilon_{it})$ in $\tilde{u}_i^{(1)}$, we have
\begin{equation}\label{doublebound}
0< R(C_k, v, \epsilon)\wedge R(C_{k+1}, v, \epsilon)\leq \sup_{\tilde{u}_i^{(1)}\in (C_k, C_{k+1}]}R(\tilde{u}_i^{(1)}, v, \epsilon)\leq R(C_k, v, \epsilon)\vee R(C_{k+1}, v, \epsilon),
\end{equation}
where the lower and upper bounds are irrelevant to $\tilde{u}_i^{(1)}$ and the results of the operators $\wedge$ and $\vee$ depend only on $sign(\epsilon_{it})$. Let
$$
G_{bk}(v, \epsilon)=R(C_k, v, \epsilon)\vee R(C_{k+1}, v, \epsilon)
$$
and
$$
G_{sk}(v, \epsilon)=R(C_k, v, \epsilon)\wedge R(C_{k+1}, v, \epsilon),
$$
which are two bivariate functions bounded by $C$ when $\|\tilde{u}_i^{(1)}\|\leq M$ and $\max_t\|v^{(1)}_t\|\leq C\log{T}/\sqrt{p}$. $G_{bk}(v^{(1)}_t, \epsilon_{it})$ and $G_{sk}(v_t^{(1)},\epsilon_{it})$ are simply the values of $G_{bk}(v, \epsilon)$ and $G_{sk}(v, \epsilon)$ at $(v, \epsilon)=(v^{(1)}_t, \epsilon_{it})$. Then (\ref{doublebound}) implies that
\begin{eqnarray}\label{doublebound2}
&&\sup_{\tilde{u}_i\in(-M,M]}\Big|\sum^T_{t=1}(R^{(1)}_{it}-E_f R^{(1)}_{it})\Big|\nonumber\\
&\leq& \max_k\Big|\sum^T_{t=1} \{G_{bk}(v^{(1)}_t, \epsilon_{it})-E_fG_{sk}(v^{(1)}_t, \epsilon_{it})\}\Big|
+\max_k\Big|\sum^T_{t=1}\{E_fG_{bk}(v^{(1)}_t, \epsilon_{it})-G_{sk}(v^{(1)}_t, \epsilon_{it})\}\Big|\nonumber\\
&\leq&\max_k\Big|\sum^T_{t=1} \{G_{bk}(v_t^{(1)},\epsilon_{it})-E_fG_{bk}(v_t^{(1)}, \epsilon_{it})\}\Big|
+\max_k\Big|\sum^T_{t=1} \{E_fG_{bk}(v_t^{(1)},\epsilon_{it})-E_fG_{sk}(v_t^{(1)}, \epsilon_{it})\}\Big|\nonumber\\
&&+\max_k\Big|\sum^T_{t=1} \{E_fG_{sk}(v_t^{(1)},\epsilon_{it})-G_{sk}(v_t^{(1)}, \epsilon_{it})\}\Big|
+\max_k\Big|\sum^T_{t=1} \{E_fG_{sk}(v_t^{(1)},\epsilon_{it})-E_fG_{bk}(v_t^{(1)}, \epsilon_{it})\}\Big|\nonumber\\
&:=&I_{1i}(\{v_t^{(1)}\},\{\epsilon_{it}\})+I_{2i}(\{v_t^{(1)}\},\{\epsilon_{it}\})+II_{1i}(\{v_t^{(1)}\},\{\epsilon_{it}\})+II_{2i}(\{v_t^{(1)}\},\{\epsilon_{it}\}).
\end{eqnarray}
A closer look at $R(C_k, v^{(1)}_t, \epsilon_{it})$ shows that $R(C_k, v^{(1)}_t, \epsilon_{it})$ is a piecewise linear monotone function in $v^{(1)}_t$ with turning points $\{v^{(1)}_t; \tau_{it}^{(1)}= \pm\epsilon_{it}\}$, and the principal term of $E_fR(C_k, v^{(1)}_t, \epsilon_{it})$ by (\ref{term2}) (i.e. $\sum^T_{t=1}h_i(0)(\bl_i^{\prime}v^{(1)}_t+C_k^{\prime}\bbf_t^0/\sqrt{T})^2$) is a quadratic function in $v^{(1)}_t$. Therefore
\begin{eqnarray}\label{uniforminv}
&&\Big|\sup_{\max_{t\leq T}\|v_t^{(1)}\|\leq C\log{T}/\sqrt{p}}V(\{v_t^{(1)}\}, \{\epsilon_{it}\})\Big|\nonumber\\
&\leq &\Big|V(\{\underline{v}^{(1)}\}, \{\epsilon_{it}\})I(\tau_{it}^{(1)}\not= \pm\epsilon_{it})\Big|+\Big|V(\{v_t^{(1)}\}, \{\epsilon_{it}\})I(\{\tau_{it}^{(1)}\}= \{\pm\epsilon_{it}\})\Big|\nonumber\\
&&+O_p\Big(\frac{M^3}{\sqrt{T}}+\frac{T(\log{T})^3}{p^{3/2}}\Big),
\end{eqnarray}
where the $O_p$ term holds uniformly in $\max_{t\leq T}\|v^{(1)}_t\|\leq C\log{T}/\sqrt{p}$, $V(\{v_t^{(1)}\}, \{\epsilon_{it}\})=\sum^T_{t=1}(R(C_k,v^{(1)}_t,\epsilon_{it})-E_fR(C_k,v^{(1)}_t,\epsilon_{it}))$ and $\underline{v}^{(1)}$ is an end point of $v_t^{(1)}$ whose coordinate components equal to $\pm sign(\epsilon_{it})\frac{\log{T}}{\sqrt{p}}$.  Because $\epsilon_{it}$ has probability density function $h_i(x)$, $E|V(\{v_t^{(1)}\}, \{\epsilon_{it}\})I(\{\tau_{it}^{(1)}\}= \{\pm\epsilon_{it}\})|=0$. Then it suffices to consider $|V(\{\underline{v}^{(1)}\}, \{\epsilon_{it}\})|$.

By Assumption 4 with $\mu_{it}=0$ and the Markov inequality,
\begin{eqnarray}\label{ub1}
P\Big\{\max_i\Big|I_{1i}(\{\underline{v}^{(1)}\}, \{\epsilon_{it}\})+II_{1i}(\{\underline{v}^{(1)}\},\{\epsilon_{it}\})\Big|>\epsilon_{p,T}\Big\}\leq CpM^{2r}e^{-\epsilon_{p,T}/\sigma_T},
\end{eqnarray}
where $\sigma_T=(\frac{M^3}{T^{1/2}}+\frac{T(\log{T})^3}{p^{3/2}})^{1/2}$.
By (\ref{term2}), Lemma \ref{lemma1} and Assumption 3,
\begin{equation}\label{ub2}
E\left\{\sup_{\max_t\|v^{(1)}_t\|\leq C\log{T}/\sqrt{p}}\max_i\Big(I_{2i}(\{v_t^{(1)}\}, \{\epsilon_{it}\})+II_{2i}(\{v_t^{(1)}\}, \{\epsilon_{it}\})\Big)\right\}\leq C\delta^{\prime}\Big(1+\frac{\sqrt{T}\log{T}}{\sqrt{p}M}\Big).
\end{equation}


Let
$$
\tilde{u}^{(1)}_i=-\frac{1}{2h_i(0)}\Big(\sum^T_{t=1}\frac{(\bW^{(0)})^{-1}\hat{\bbf}_t^{(1)}\hat{\bbf}_t^{(1)\prime}(\bW^{(0)\prime})^{-1}}{T}\Big)^{-1}\sum^T_{t=1}D_{it}\frac{(\bW^{(0)})^{-1}\hat{\bbf}_t^{(1)}}{\sqrt{T}}.
$$
(\ref{dec2})-(\ref{ub2}) show that
\begin{eqnarray}\label{bh}
&&\sum^T_{t=1}G_f(\tilde{u}_i, v^{(1)}_t)\nonumber\\
&=&\sum^T_{t=1}D_{it}\tau^{(1)}_{it}+\sum^T_{t=1}h_i(0)\Big(\tilde{u}^{\prime}_i(\bW^{(0)})^{-1}\hat{\bbf}^{(1)}_t/\sqrt{T}+\bl_i^{0\prime}(\bW^{(0)})^{-1}v^{(1)}_t\Big)^2\nonumber\\
&&+O_p\Big(\epsilon_{p,T}+\frac{M^3}{\sqrt{T}}+\frac{T(\log{T})^3}{p^{3/2}}\Big)+o_p\Big(1+\frac{\sqrt{T}\log{T}}{\sqrt{p}M}\Big)\nonumber\\
&=& (\tilde{u}_i-\tilde{u}^{(1)}_i)^{\prime}\Big(\sum^T_{t=1}h_i(0)\frac{\hat{\bbf}^{(1)}_t\hat{\bbf}_t^{(1)\prime}}{T}\Big)(\tilde{u}_i-\tilde{u}^{(1)}_i)\nonumber\\
&&+\sum^T_{t=1}h_i(0)\Big(\bl_i^{0\prime}(\bW^{(0)})^{-1}v^{(1)}_t\Big)^2+\sum^T_{t=1}D_{it}\bl_i^{0\prime}(\bW^{(0)})^{-1}v^{(1)}_t+O_p(\epsilon_{p,T})+o_p\Big(1+\sqrt{\frac{T}{p}}M\log{T}\Big)\nonumber\\
&&+O_p\Big(\frac{M^3}{\sqrt{T}}+\frac{T(\log{T})^3}{p^{3/2}}\Big)+O_p\Big(\frac{M}{\sqrt{p}}\Big),
\end{eqnarray}
where the $o_p$ and $O_p$ terms hold uniformly in $\{\max_i\|\tilde{u}_i|\leq M, \max_t\|v_t^{(1)}\|\leq C\log{T}/\sqrt{p}\}$. Without affecting the asymptotics below, we restrict that $|o_p(1+\sqrt{T/p}M\log{T})|\leq \epsilon(1+\sqrt{T/p}M\log{T})$ for some arbitrarily small $\epsilon>0$.
Let
$$
\sigma_{p, T}=:\frac{M^3}{\sqrt{T}}+\frac{T(\log{T})^3}{p^{3/2}}+\frac{M}{\sqrt{p}}.
$$
Next, we show that $\tilde{u}_i$ is around $\tilde{u}^{(1)}_i$. We first restrict our study on the set $\{\max_i\|\tilde{u}^{(1)}_i\|\leq M\}$. Let $\tilde{u}_i=\beta_i e_i+\tilde{u}^{(1)}_i$ for $\beta_i>\delta$ and $u_i^*=\delta e_i+\tilde{u}^{(1)}_i$. By Lemma \ref{lemma3}, (\ref{bh}), the convexity of $G_f(\tilde{u}_i, \hat{v}^{(1)}_t)$ in $\tilde{u}_i$ for fixed $\hat{v}^{(1)}_t$, $\bbf_t^0$ and $\bW^{(0)}$,
\begin{eqnarray}\label{convex}
&&\sum^T_{t=1}\Big(G_f(\tilde{u}_i, \hat{v}_t^{(1)})-G_f(\tilde{u}^{(1)}_i, \hat{v}^{(1)}_t)\Big)> \frac{\beta_i}{\delta}\sum^T_{t=1}\Big(G_f(u_i^*, \hat{v}_t^{(1)})-G_f(\tilde{u}^{(1)}_i, \hat{v}_t^{(1)})\Big)\nonumber\\
&=&\beta_i\delta e_i^{\prime}\Big(\sum^T_{t=1}h_i(0)\frac{\hat{\bbf}^{(1)}_t\hat{\bbf}_t^{(1)\prime}}{T}\Big)e_i+\Big\{O_p(\epsilon_{p,T})+O_p(\sigma_{p, T})+o_p\Big(1+\sqrt{\frac{T}{p}}M\log{T}\Big)\Big\}\frac{\beta_i}{\delta},
\end{eqnarray}
where the $o_p$ and $O_p$ terms hold uniformly in $\{\max_i\|\tilde{u}_i|\leq M\}$. Now, we prove the first equation by setting $M=C\{\log{p}+({\epsilon T}/{p})^{1/2}\}\log{T}$, $\delta=C_1\{\log{p}+({\epsilon T}/{p})^{1/2}\}\log{T}$ and $\epsilon_{p,T}=\delta^2/M^*$ for large enough $M^*$, $C$ and $C_1$. By the Chebyshev inequality,
\begin{eqnarray}\label{Chebyshev1}
&&P\Big\{\max_{1\leq i\leq p}\|\sqrt{T}\hat{u}^{(1)}_i-\tilde{u}^{(1)}_i\|>\delta\Big\}\nonumber\\
&=&P\Big\{\max_{1\leq i\leq p}\|\sqrt{T}\hat{u}^{(1)}_i-\tilde{u}^{(1)}_i\|>\delta, \max_i\|\tilde{u}^{(1)}_i\|\leq M, \max_t\|\hat{v}_t^{(1)}\|\leq C\log{T}/\sqrt{p}\Big\}\nonumber\\
&&+P\Big\{\max_i\|\tilde{u}^{(1)}_i\|> M\Big\}+P\Big\{\max_t\|\hat{v}_t^{(1)}\|> C\log{T}/\sqrt{p}\Big\}.
\end{eqnarray}
Assumption 4 and the Markov inequality show that
\begin{equation}\label{barui}
P\Big\{\max_i\|\tilde{u}^{(1)}_i\|> M\Big\}\leq Cpe^{-M}=o(1).
\end{equation}
Lemma \ref{lemma1}, Assumption 4-5, the Bonferroni inequality and the Markov inequality prove that
\begin{equation}\label{hatvt}
P\Big\{\max_t\|\hat{v}^{(1)}_t\|> C\log{T}/\sqrt{p}\Big\}=o(1).
\end{equation}
Assumptions 4-5 and (\ref{convex}) yield
\begin{eqnarray}\label{barui1}
&&P\Big\{\max_{1\leq i\leq p}\|\sqrt{T}\hat{u}^{(1)}_i-\tilde{u}^{(1)}_i\|>\delta, \max_i\|\tilde{u}^{(1)}_i\|\leq M, \max_t\|\hat{v}_t^{(1)}\|\leq C\log{T}/\sqrt{p}\Big\}\nonumber\\
&\leq & P\Big\{\beta_i\delta e_i^{\prime}(\sum^T_{t=1}h_i(0)\frac{\hat{\bbf}^{(1)}_t\hat{\bbf}_t^{(1)\prime}}{T})e_i+\Big[O_p(\epsilon_{p,T})+O_p(\sigma_{p, T})-\epsilon(1+\sqrt{T/p})\Big]\frac{\beta_i}{\delta}<0\Big\}\nonumber\\
&=&P\Big\{ e_i^{\prime}(\sum^T_{t=1}h_i(0)\frac{\hat{\bbf}^{(1)}_t\hat{\bbf}_t^{(1)\prime}}{T})e_i+\delta^{-2}\Big[O_p(\epsilon_{p,T})+O_p(\sigma_{p, T})-\epsilon(1+\sqrt{T/p}M\log{T})\Big]<0\Big\}\nonumber\\
&=&P\Big\{ e_i^{\prime}(\sum^T_{t=1}h_i(0)\frac{\hat{\bbf}^{(1)}_t\hat{\bbf}_t^{(1)\prime}}{T})e_i-\epsilon<0\Big\}=o(1),
\end{eqnarray}
by letting $p, T\rightarrow\infty$ first and then $\epsilon\rightarrow 0$, where the last equality is due to Assumption 3 and the identifiability condition

\begin{equation}\label{equ:identifiability}
\bL^{\prime}\bL/p \ \text{is  diagonal  and } \frac{1}{T}\sum^T_{t=1}\bbf_t\bbf_t^{\prime}=\Ib_r.
\end{equation}

(\ref{Chebyshev1})-(\ref{barui1}) and Lemma \ref{lemma1} prove that
$$
\hat{u}^{(1)}_i=\frac{1}{\sqrt{T}}\tilde{u}^{(1)}_i+O_{p}\Big(\frac{\log{p}+(\sqrt{\epsilon T/p})}{\sqrt{T}}\Big)=O_p\Big(\frac{\log{p}}{\sqrt{T}}\Big)+o_p\Big(\frac{1}{\sqrt{p}}\Big),
$$
where the $O_p$ and $o_p$ terms hold uniformly in $i\leq p$, and the last equality is due to the condition on $D_{it}$ in Assumption 4. This proves the first equation of Lemma \ref{lemma2}.

To prove the Barhadur representation for $\hat{u}_{i}^{(1)}$ in Lemma \ref{lemma2}, let $\delta$ and $\epsilon_{p, T}$ be two arbitrarily small constants, and $M=C\log{p}$. Because we further have the condition that $((\log{p})^2(\log{T})^2+(\log{T})^3/\sqrt{p}){T}/{p}+{\log^5{p}}/{\sqrt{T}}=o(1)$ and the condition on $\epsilon_{it}$'s in Assumption 4, the probability in (\ref{ub1}) is $o(1)$ and (\ref{convex})-(\ref{barui1}) are still true, which proves
$$
P\Big\{\max_{i\leq p}\|\sqrt{T}\hat{u}_i^{(1)}-\tilde{u}_i^{(1)}\|>\delta\Big\}\rightarrow 0.
$$
To complete the proof of the second equation of Lemma \ref{lemma2}, it suffices to prove
\begin{equation}\label{t-bar}
P\Big\{\max_i\|\tilde{u}_i^{(1)}-\overline{u}_i^{(1)}\|>\delta\Big\}\rightarrow 0.
\end{equation}
By Lemma \ref{lemma3} and the boundedness of $D_{it}$,
\begin{equation}\label{t-bar-1}
\max_i\Big\{\frac{1}{\sqrt{T}}\sum^T_{t=1}D_{it}(\bW^{(0)})^{-1}(\hat{\bbf}_t^{(1)}-\bW^{(0)}\bbf_t^0)\Big\}=O_p\Big(\sqrt{T}\log{T}/\sqrt{p}\Big)=o_p(1).
\end{equation}
\begin{eqnarray}\label{t-bar-3}
&&\frac{1}{T}\sum^T_{t=1}\Big\{(\bW^{(0)})^{-1}\hat{\bbf}_t^{(1)}-\bbf_t^0\Big\}\hat{\bbf}_t^{(1)\prime}(\bW^{(0)\prime})^{-1}\nonumber\\
&\leq& C\max_t\Big\|\hat{\bbf}_t^{(1)}-\bW^{(0)}\bbf_t^0\Big\|\frac{1}{T}\sum^T_{t=1}\|\hat{\bbf}_t^{(1)\prime}\|\|(\bW^{(0)\prime})^{-1}\|=O_p\Big(\frac{\log{T}}{\sqrt{p}}\Big),\nonumber\\
&&\frac{1}{T}\sum^T_{t=1}\bbf_t^0\Big\{\hat{\bbf}_t^{(1)\prime}(\bW^{(0)\prime})^{-1}-\bbf_t^{0\prime}\Big\}\nonumber\\
&\leq & C\max_t\Big\|\hat{\bbf}_t^{(1)}-\bW^{(0)}\bbf_t^0\Big\|\frac{1}{T}\sum^T_{t=1}\|\bbf_t^{0\prime}\|\|(\bW^{(0)\prime})^{-1}\|=O_p\Big(\frac{\log{T}}{\sqrt{p}}\Big).
\end{eqnarray}
Assumption {\bf 1'} and (\ref{t-bar-1})-(\ref{t-bar-3}) prove (\ref{t-bar}).

\end{proof}

Next, we turn to the $k$-th ($k\geq 2$) update of $\bbf_t$ and its corresponding distance to $\bW^{(k-1)}\bbf_t^0$ with updated $\bl_i^{(k-1)}$'s. Similar to $\bW^{(0)}$ defined in the proof of Lemma \ref{lemma1}, here we define
$$
\bW^{(k)}=:\Big\{\sum^p_{i=1}h_i(\xi_{it}^{(k)})\hat{\bl}_i^{(k)}\hat{\bl}_i^{(k)\prime}\Big\}^{-1}\sum^{p}_{i=1}h_i(\xi_{it}^{(k)})\hat{\bl}_i^{(k)}\bl_i^{0\prime},
$$
where for some $\theta^{(k)}_{it}\in [0, 1]$, $\xi_{it}^{(k)}=\theta^{(k)}_{it}(\hat{\bl}_i^{(k)\prime}\bW^{(k)}-\bl_i^{0\prime})\bbf_t^0$  is some variable between $0$ and  $(\hat{\bl}_i^{(k)\prime}\bW^{(k)}-\bl_i^{0\prime})\bbf_t^0$. Notice that $\bW^{(k)}$ depends only on $\hat{\bl}_i^{(k)}$'s, $\bl_i^0$'s and $\bbf_t^0$.
From now on, we define $\hat{v}^{(k)}_t=\hat{\bbf}_t^{(k)}-\bW^{(k-1)}\bbf_t^0$,  and $v^{(k)}_t=\bbf_t^{(k)}-\bW^{(k-1)}\bbf_t^0$, and $\hat{u}^{(k)}_i=\bW^{(k-1)\prime}\hat{\bl}_i^{(k)}-\bl_i^0$ and $u^{(k)}_i=\bW^{(k-1)\prime}\bl_i^{(k)}-\bl_i^0$.
Let $A_k$ and $B_k$ be the sets of samples so that $\max_t|\hat{v}_t^{(k)}|\leq C\frac{\log{T}}{\sqrt{p}}$ and $\max_i|\hat{u}_i^{(k)}|\leq C({\log{p}}/{\sqrt{T}}+{\epsilon}/{\sqrt{p}})$, respectively.

\begin{lemma}\label{lemma3}
For fixed $\hat{\bl}^{(k-1)}_i$'s and given $\bbf_t^0$, under Assumptions 1-5,
$$
P\left(\max_t\Big\|\sqrt{p}\hat{v}^{(k)}_t-\overline{v}^{(k)}_t\Big\|> \delta\right)\rightarrow 0,
$$
for any $\delta>0$, where
$$
\overline{v}^{(k)}_t=-\frac{1}{2}\Big\{\frac{1}{p}\sum^p_{i=1}h_i(0)\hat{\bl}^{(k-1)}_i\hat{\bl}_i^{(k-1)\prime}\Big\}^{-1}\sum^p_{i=1}D_{it}\frac{\hat{\bl}_i^{(k-1)}}{\sqrt{p}},
$$
\end{lemma}

\begin{proof}
The proof of Lemma \ref{lemma3} is similar to that of Lemma \ref{lemma1} except for updating $\bl_i^{(0)}$ and $\bW^{(0)}$ by $\hat{\bl}_i^{(k-1)}$ and $\bW^{(k-1)}$, respectively, and noting that $\hat{u}_i^{(k-1)}$'s are in $B_{k-1}$. Let $\tilde{\epsilon}^{(k)}_{it}=\epsilon_{it}-u_i^{(k)\prime}\bbf_t^0$ and $\tilde{\tau}_{it}^{(k)}=p^{-1/2}\bl_i^{\prime}\tilde{v}_t^{(k)}$. Indeed, we show that the expansion of $\sum^p_{i=1}g_l(u_i^{(k-1)}, v_t^{(k)})$ in (\ref{exp1}) with $(u_i^{(0)}, v_t)$ there replaced by $(u_i^{(k-1)}, v_t^{(k)})$ holds uniformly in $\{\max_t\|v_t^{(k)}\|\leq C\log{T}/\sqrt{p}, \max_i\|u_i^{(k-1)}\|\leq C(\log{p}/\sqrt{T}+\epsilon/\sqrt{p})\}$.
First, (\ref{eR}) with $\tilde{v}_t$ there replaced by $\tilde{v}_t^{(k)}=\sqrt{p}v_t^{(k)}$ holds uniformly in $\{\max_i\|u_i^{(k-1)}\|\leq C(\log{p}/\sqrt{T}+\epsilon/\sqrt{p})\}$. Then it suffices to prove that
\begin{equation}\label{R-eR-k}
\sup_{\max_i\|u_i^{(k-1)}\|\leq C(\log{p}/\sqrt{T}+\epsilon/\sqrt{p})}\sum^p_{i=1}\Big(\tilde{R}_{it}^{(k-1)}-E_f\tilde{R}_{it}^{(k-1)}\Big)=o_p(1),
\end{equation}
with $\tilde{R}_{it}^{(0)}$ in (\ref{R-eR}) replaced by $\tilde{R}_{it}^{(k-1)}$ which is similarly defined as $\tilde{R}_{it}^{(0)}$ except for replacing $\bl_i^{(0)}$ and $\bW^{(0)}$ by $\bl_i^{(k-1)}$ and $\bW^{(k-1)}$, respectively. To this end, replace $(u_i^{(0)}, \bl_i^{(0)})$ in (\ref{Imean}) by $(u_i^{(k-1)}, \bl_i^{(k-1)})$, one easily shows that
\begin{equation}\label{Imeank}
\sup_{\max_i\|u_i^{(k-1)}\|\leq C(\log{p}/\sqrt{T}+\epsilon/\sqrt{p})}\Big[I_{2t}+II_{1t}\Big]\leq C\delta^{\prime}.
\end{equation}
Define
$$
U\Big(\{u_i^{(k-1)}\}, \{\epsilon_{it}\}\Big)=\sum^p_{i=1}\Big(R\Big(u_i^{(k-1)},C_k, \epsilon_{it}\Big)-E_fR\Big(u_i^{(k-1)},C_k, \epsilon_{it}\Big)\Big).
$$
We see that
$$
\sum^p_{i=1}E_fR\Big(u_i^{(k-1)},C_k, \epsilon_{it}\Big)=\sum^p_{i=1}h_i(0)\Big(\tilde{\tau}_{it}^{(k-1)}\Big)^2+o_p(1),
$$
where $o_p(1)$ holds uniformly in $\{\max_t\|v_t^{(k)}\|\leq C\log{T}/\sqrt{p}, \max_i\|u_i^{(k-1)}\|\leq C(\log{p}/\sqrt{T}+\epsilon/\sqrt{p})\}$. Notice that $R(u_i^{(k-1)},C_k, \epsilon_{it})$ is a piecewise linear function in $\tilde{\tau}_{it}^{(k-1)}$ with turning points $\pm\tilde{\epsilon}_{it}^{(k-1)}$ while the principal term of $E_fR(u_i^{(k-1)},C_k, \epsilon_{it})$ is a quadratic function in $\tilde{\tau}_{it}^{(k-1)}$. Then $\sup_{\max_i\|u_i^{(k)}\|\leq C(\log{p}/\sqrt{T}+\epsilon/\sqrt{p})}U(\{u_i^{(k-1)}\}, \{\epsilon_{it}\})$ is achieved when $\tilde{\tau}_{it}^{(k-1)}=\tilde{\epsilon}_{it}^{(k-1)}$ or $\tilde{\tau}_{it}^{(k-1)}$ equals an end point. Because the probability density function exists for $\epsilon_{it}$, $E_f|U(\{u_i^{(k-1)}\}, \{\epsilon_{it}\})|I\{\tilde{\tau}_{it}^{(k-1)}=\tilde{\epsilon}_{it}^{(k-1)}\}=0$. When $\tilde{\tau}_{it}^{(k-1)}$, $R(u_i^{(k-1)}, C_k, \epsilon_{it})$ is a piecewise linear function in $u_i^{(k-1)}$, and then $$\sup_{\max_i\|u_i^{(k-1)}\|\leq C(\log{p}/\sqrt{T}+\epsilon/\sqrt{p})}U\Big(\{u_i^{(k-1)}\}, \{\epsilon_{it}\}\Big)
$$
is achieved when $u_i^{(k-1)}=\pm \epsilon_{it}-C\log{T}/\sqrt{p}$, $u_i^{(k-1)}=\pm \epsilon_{it}$, or $\|u_i^{(k-1)}\|=C(\log{p}/\sqrt{T}+\epsilon/\sqrt{p})$ whose solution is denoted by $\underline{u}^{(k-1)}$ which is independent of $i$. For the first the two cases, the probability of the two events are zero, hence it is enough to consider $U(\{\underline{u}^{(k-1)}\}, \{\epsilon_{it}\})$ which, similar to (\ref{Ivar}), is $o_p(1)$ due to Assumptions {\bf 1}' and 3. This completes the proof of (\ref{R-eR-k}) and hence the expansion of $\sum^p_{i=1}g_l(u_i^{(k-1)}, v_t^{(k)})$ in (\ref{exp1}) with $(u_i^{(0)}, v_t, \tilde{D}_{it}^{(0)}, \bW^{(0)}, \bl_i^{(0)})$ there replaced by $(u_i^{(k-1)}, v_t^{(k)}, \tilde{D}_{it}^{(k)}, \bW^{(k)}, \bl_i^{(k)})$ holds uniformly in $\{\max_t\|v_t^{(k)}\|\leq C\log{T}/\sqrt{p}, \max_i\|u_i^{(k-1)}\|\leq C(\log{p}/\sqrt{T}+\epsilon/\sqrt{p})\}$.

Following exactly the same lines as in the remaining proof of Lemma \ref{lemma1} (the lines below (\ref{exp1})), we have
\begin{eqnarray}\label{Chebyshevk}
&&P\Big\{\max_t\|\sqrt{p}\hat{v}^{(k)}_t-\overline{v}^{(k)}_t\|>\delta\Big\}\nonumber\\
&\leq & P\Big\{\sum^p_{i=1}\Big[G_l(\hat{u}_i^{(k-1)},\breve{v}^{(k)}_t)-G_l(\hat{u}_i^{(k-1)},\overline{v}^{(k)}_t)\Big]<0, \max_t\|\overline{v}^{(k)}_t\|\leq C\log{T}-\delta,\nonumber\\
&& \max_i\|\hat{u}_i^{(k-1)}\|\leq C(\log{p}/\sqrt{T}+\epsilon/\sqrt{p})\Big\} + P\Big\{\max_t\|\overline{v}^{(k)}_t\|> C\log{T}-\delta\Big\}\nonumber\\
&& +P\Big\{\max_i\|\hat{u}_i^{(k-1)}\|> C\Big(\frac{\log{p}}{\sqrt{T}}+\frac{\epsilon}{\sqrt{p}}\Big)\Big\} \leq C\epsilon.
\end{eqnarray}
Notice that $\max_{i}|\tilde{D}_{it}^{(k)}-D_{it}|=o_p(1)$ due to the restriction that
$$
\max_i\|\hat{u}_i^{(k-1)}\|\leq C\Big(\log{p}/\sqrt{T}+\epsilon/\sqrt{p}\Big).
$$
This shows that $\tilde{D}_{it}^{(k)}-E_f\tilde{D}_{it}^{(k)}$ can be replaced by $D_{it}$ in the definition of $\overline{v}_t^{(1)}$. This completes the proof of the lemma.

\end{proof}

\begin{lemma}\label{lemma4}
Under Assumptions 1-5,
$
A_k\subseteq B_k\subseteq A_{k+1}
$
with probability approaching one.
\end{lemma}
\begin{proof}
The proof of Lemma \ref{lemma3} shows that once $\hat{u}_i^{(k)}$'s enter $B_k$ and satisfy Assumptions 1-5, $\hat{v}_t^{(k+1)}$'s will satisfy the condition of $A_{k+1}$ on a subsample space $S_1$ with probability larger than $1-\epsilon/2$ for some arbitrarily small $\epsilon>0$. The proof of Lemma \ref{lemma2} shows that once $\hat{v}_t^{(k)}$'s enter $A_k$ and satisfy Assumptions 1-5, $\hat{u}_i^{(k)}$'s will satisfy the condition of $B_{k}$ on a subsample space $S_2$ with probability larger than $1-\epsilon/2$. Then $P(S_1\cap S_2)\geq 1-\epsilon$.

\end{proof}

Next, we show that $\bW^{(k)}$ is close to $\bW^{(k-1)}$ when $k\geq 1$ and hence as implied by Lemma \ref{lemma2}, $\hat{\bl}_i^{(k)\prime}\bW^{(k-1)}-\bl_i^{0\prime}$ is close to zero, which further shows that $\bW^{(k)}$ and $\tilde{\bW}_0$ are close enough.

\begin{lemma}\label{lemma5}
Under Assumptions 1-5,
\begin{eqnarray}
\Big\|\bW^{(k)}-\bW^{(k-1)}\Big\|&=& O_p(\frac{\log{p}}{\sqrt{T}})+o_p\Big(\frac{1}{\sqrt{p}}\Big),\label{w1-w0} \\
\max_i\Big\|\bW^{(k)\prime}\hat{\bl}_i^{(k)}-\bl_i^0\Big\|&=&O_p\Big(\frac{\log{p}}{\sqrt{T}}\Big)+o_p\Big(\frac{1}{\sqrt{p}}\Big),\label{l-l0} \\
\Big\|\bW^{(k)}-\tilde{\bW}_0\Big\|&= & O_p\Big(\frac{\log{p}}{\sqrt{T}}\Big)+o_p\Big(\frac{1}{\sqrt{p}}\Big).\label{w1-w00}
\end{eqnarray}

If further $[(\log{p})^2(\log{T})^2+(\log{T})^3/\sqrt{p}]\frac{T}{p}+\frac{\log^5{p}}{\sqrt{T}}=o(1)$,
\begin{eqnarray}
\Big\|\bW^{(k)}-\bW^{(k-1)}\Big\|&=& o_p\Big(\frac{1}{\sqrt{T}}\Big),\label{1-w1-w0} \\
\max_i\Big\|\bW^{(k)\prime}\hat{\bl}_i^{(k)}-\bl_i^0-\frac{1}{\sqrt{T}}\overline{u}_i^{(k)}\Big\|&=&o_p\Big(\frac{1}{\sqrt{T}}\Big),\label{1-l-l0} \\
\Big\|\bW^{(k)}-\tilde{\bW}_0\Big\|&= & O_p\Big(\frac{1}{\sqrt{T}}\Big).\label{1-w1-w00}
\end{eqnarray}
\end{lemma}

\begin{proof}

Lemma \ref{lemma4} shows that
\begin{equation}\label{w1-w0-1}
\frac{1}{p}\sum_i\hat{\bl}_i^{(k)}\hat{\bl}_i^{(k)\prime}h_i(\xi_{it}^{(k)})=\Big(\bW^{(k-1)\prime}\Big)^{-1}\frac{1}{p}\sum_i\bl_i^{0}\bl_i^{0\prime}h_i(\xi_{it}^{(k)})\Big(\bW^{(k-1)}\Big)^{-1}+O_p\Big(\frac{\log{p}}{\sqrt{T}}\Big)+o_p\Big(\frac{1}{\sqrt{p}}\Big),
\end{equation}
and
\begin{equation}\label{w1-w0-2}
\frac{1}{p}\sum_i\hat{\bl}_i^{(k)}\bl_i^{0\prime}h_i(\xi_{it}^{(k)})=\Big(\bW^{(k-1)\prime}\Big)^{-1}\frac{1}{p}\sum_i\bl_i^{0}\bl_i^{0\prime}h_i(\xi_{it}^{(k)})+O_p\Big(\frac{\log{p}}{\sqrt{T}}\Big)+o_p\Big(\frac{1}{\sqrt{p}}\Big).
\end{equation}
Combining (\ref{w1-w0-1}) and (\ref{w1-w0-2}) proves (\ref{w1-w0}). (\ref{w1-w0}) and Lemma \ref{lemma4} prove that
\begin{eqnarray}\label{p-w1-w0}
&&\bW^{(k)\prime}\hat{\bl}_i^{(k)}-\bl_i^0\nonumber\\
&=& \bW^{(k-1)\prime}\hat{\bl}_i^{(k)}-\bl_i^0+\Big\{\bW^{(k)\prime}-\bW^{(k-1)\prime}\Big\}\Big(\hat{\bl}_i^{(k)}-(\bW^{(k-1)\prime})^{-1}\bl_i^0+(\bW^{(k-1)\prime})^{-1}\bl_i^0\Big)\\
&=&O_p\Big(\frac{\log{p}}{\sqrt{T}}\Big)+o_p\Big(\frac{1}{\sqrt{p}}\Big),\nonumber
\end{eqnarray}
where the $O_p$ and $o_p$ terms hold uniformly in $i\leq p$. The only difference between $\bW^{(k)}$ and $\bW_0$ is the difference between $h_i(\xi_{it}^{(k)})$ and $h_i(0)$, then (\ref{w1-w00}) is a straightforward result of (\ref{l-l0}) and the property $\max_i\|\xi_{it}^{(k)}\|\leq \max_i\|\bW^{(k)\prime}\hat{\bl}_i^{(k)}-\bl_i^0\|=O_p({\log{p}}/{\sqrt{T}})+o_p({1}/{\sqrt{p}})$.

If further $[(\log{p})^2(\log{T})^2+(\log{T})^3/\sqrt{p}]\frac{T}{p}+\frac{\log^5{p}}{\sqrt{T}}=o(1)$, the proof of Lemma \ref{lemma2} and Lemma \ref{lemma4} show that
\begin{equation}\label{star}
\hat{u}_i^{(k)}=\frac{1}{\sqrt{T}}\overline{u}_i^{(k)}+o_p\Big(\frac{1}{\sqrt{T}}\Big),
\end{equation}
where the $o_p$ term holds uniformly in $i\leq p$. Equation (\ref{star}) together with the temporal and cross-section weak dependence condition on $D_{it}$'s in Assumption 4 proves that
\begin{eqnarray}\label{hui}
&&\frac{1}{p}\sum^p_{i=1}h_i(\xi_{it}^{(k)})\Big\{\hat{\bl}_i^{(k)}-(\bW^{(k-1)\prime})^{-1}\bl_i^0\Big\}\bl_i^{0\prime}\nonumber\\
&=&\Big(\bW^{(k-1)\prime}\Big)^{-1}\frac{1}{p\sqrt{T}}\sum^p_{i=1}h_i(\xi_{it}^{(k)})\overline{u}_i^{(k)}\bl_i^{0\prime}+o_p\Big(\frac{1}{\sqrt{T}}\Big)=O_p\Big(\frac{1}{\sqrt{pT}}\Big)+o_p\Big(\frac{1}{\sqrt{T}}\Big),
\end{eqnarray}
\begin{eqnarray}\label{hui1}
&&\frac{1}{p}\sum^p_{i=1}h_i(\xi_{it}^{(k)})\Big\{\hat{\bl}_i^{(k)}-(\bW^{(k-1)\prime})^{-1}\bl_i^0\Big\}\hat{\bl}_i^{(k)\prime}\nonumber\\
&=&\Big(\bW^{(k-1)\prime}\Big)^{-1}\frac{1}{p\sqrt{T}}\sum^p_{i=1}h_i(\xi_{it}^{(k)})\overline{u}_i^{(k)}\hat{\bl}_i^{(k)\prime}+o_p\Big(\frac{1}{\sqrt{T}}\Big)=O_p\Big(\frac{1}{\sqrt{pT}}\Big)+o_p\Big(\frac{1}{\sqrt{T}}\Big),
\end{eqnarray}
and
\begin{eqnarray}\label{hui2}
&&\frac{1}{p}\sum^p_{i=1}h_i(\xi_{it}^{(k)})\hat{\bl}_i^{(k)}\Big\{\hat{\bl}_i^{(k)}-(\bW^{(k-1)\prime})^{-1}\bl_i^0\Big\}^{\prime}\nonumber\\
&=&\Big(\bW^{(k-1)\prime}\Big)^{-1}\frac{1}{p\sqrt{T}}\sum^p_{i=1}h_i(\xi_{it})\hat{\bl}_i^{(k)}\overline{u}_i^{(k)\prime}+o_p\Big(\frac{1}{\sqrt{T}}\Big)=O_p\Big(\frac{1}{\sqrt{pT}}\Big)+o_p\Big(\frac{1}{\sqrt{T}}\Big),
\end{eqnarray}
where the $O_p({1}/{\sqrt{pT}})$ term is due to $\overline{u}_i^{(k)}/\sqrt{T}$ and Assumption 3, and $o_p({1}/{\sqrt{T}})$ is due to (\ref{star}).
(\ref{hui})-(\ref{hui2}) prove (\ref{1-w1-w0}). (\ref{p-w1-w0}) and (\ref{1-w1-w0}) prove (\ref{1-l-l0}). For $\tilde{\theta}^{(k)}_{it}\in [0, 1]$,
\begin{equation}\label{starstar}
h_i(\xi_{it}^{(k)})-h_i(0)=\dot{h}_i(\tilde{\theta}^{(k)}_{it}\xi_{it}^{(k)})\theta_{it}^{(k)}(\hat{\bl}_i^{(k)\prime}\bW^{(k)}-\bl_i^{0\prime})\bbf_t^0.
\end{equation}
This together with (\ref{1-l-l0}), the boundedness of $\dot{h}_i(x)$, and the temporal and cross-section weak dependence condition on $D_{it}$'s in Assumption 3 proves that
\[
\frac{1}{p}\sum^p_{i=1}\{h_i(\xi_{it}^{(k)})-h_i(0)\}\hat{\bl}_i^{(k)}\bl_i^{0\prime}=O_p(1/\sqrt{T}), \hspace{1em}
\frac{1}{p}\sum^p_{i=1}\{h_i(\xi_{it}^{(k)})-h_i(0)\}\hat{\bl}_i^{(k)}\hat{\bl}_i^{(k)\prime}=O_p(1/\sqrt{T}).
\]

\end{proof}

\begin{corollary}\label{cor2}
Under the conditions in Lemma \ref{lemma3},
$$
\max_i\|\tilde{\bW}_0^{\prime}\hat{\bl}_i^{(k)}-\bl_i^0\|=O_p\Big(\frac{\log{p}}{\sqrt{T}}+o_p\Big(\frac{1}{\sqrt{p}}\Big)\Big).
$$
Moreover, if further $[(\log{p})^2(\log{T})^2+(\log{T})^3/\sqrt{p}]\frac{T}{p}+\frac{\log^5{p}}{\sqrt{T}}=o(1)$,
$$
\max_i\Big\|\bW^{(k)\prime}\hat{\bl}_i^{(k)}-\bl_i^0-\frac{1}{\sqrt{T}}\overline{u}_i^{(k)}\Big\|=o_p\Big(\frac{1}{\sqrt{T}}\Big).
$$
\end{corollary}
\begin{proof}
Corollary \ref{cor2} is a direct result of (\ref{w1-w00}) and (\ref{1-w1-w00}) in Lemma \ref{lemma5} by simply replacing $\bW^{(k)}$ in (\ref{l-l0}) and (\ref{1-l-l0}) by $\tilde{\bW}_0$.
\end{proof}

\begin{corollary}\label{cor3}
Under Assumptions 1-5,
$$
\hat{\bbf}_t^{(k+1)}-\tilde{\bW}_0\bbf_t^0=O_p\Big(\frac{\log{p}}{\sqrt{T}}\Big)+O_p\Big({\frac{1}{\sqrt{p}}}\Big).
$$
Moreover, if further ${p\log^2{p}}/{T}=o(1)$, then $$\hat{\bbf}_t^{(k+1)}-\tilde{\bW}_0\bbf_t^0=\frac{1}{\sqrt{p}}\overline{v}_t^{(k+1)}+o_p\Big(\frac{1}{\sqrt{p}}\Big)$$.
\end{corollary}
\begin{proof}
By Lemmas \ref{lemma3}-\ref{lemma5},
\begin{eqnarray*}
\hat{\bbf}_t^{(k+1)}-\tilde{\bW}_0\bbf_t^0=\hat{\bbf}_t^{(k+1)}-\bW^{(k)}\bbf_t^0+(\bW^{(k)}-\tilde{\bW}_0)\bbf_t^0=O_p\Big(\frac{1}{\sqrt{p}}\Big)+O_p\Big(\frac{\log{p}}{\sqrt{T}}\Big)+o_p\Big(\frac{1}{\sqrt{p}}\Big).
\end{eqnarray*}
If ${p\log^2{p}}/{T}=o(1)$,
the above equation demonstrates that
\begin{eqnarray*}
\hat{\bbf}_t^{(k+1)}-\tilde{\bW}_0\bbf_t^0=\hat{\bbf}_t^{(k+1)}-\bW^{(k)}\bbf_t^0+o_p\Big(\frac{1}{\sqrt{p}}\Big),
\end{eqnarray*}
and hence by Lemma \ref{lemma4}, $\hat{\bbf}_t^{(k+1)}-\tilde{\bW}_0\bbf_t^0=\overline{v}_t^{(k+1)}/\sqrt{p}+o_p({1}/{\sqrt{p}})$.
\end{proof}

\section{Proof of Main Theorems}\label{sec:B}

\noindent{\bf Proof of Theorem 1} Theorem 1 is a direct consequence of Corollaries 1 and 2. The identifiability condition (\ref{equ:identifiability}) and the first equation of Theorem 1 show that
\begin{eqnarray*}
\Ib_r=\frac{1}{T}\sum^T_{t=1}\tilde{\bbf}_t\tilde{\bbf}_t^{\prime}=\frac{1}{T}\sum^T_{t=1}\tilde{\bW}_0\bbf_t^0\bbf_t^{0\prime}\tilde{\bW}_0^{\prime}+o_p(1)=\tilde{\bW}_0\tilde{\bW}_0^{\prime}+o_p(1).
\end{eqnarray*}
This together with (\ref{1-w1-w00}) proves $\bW^{(k)}\bW^{(k)\prime}=\Ib_r+o_p(1)$.

\vspace{2em}

{
\noindent{\bf Proof of Theorem 2 and Theorem 3}
We start with the first iterative step, which can be divided into the two following parts: 1) given the initial $\bl_i^{(0)}$, estimate the factor scores and get $\hat\bbf_t^{(1)}$; 2) given $\hat\bbf_t^{(1)}$, estimate the loadings and get $\hat\bl_i^{(1)}$. Without loss of generality, we let $\tau=0.5$ in the following analysis.
\vspace{1em}

\noindent\textbf{Part 1)}:
By Assumption 1''(1), there always exists a $p\times (r_{\max}-r)$ matrix, denoted as $\bL^{-r}$, such that the eigenvalues of $p^{-1}{\bL^{(0)\prime}(r_{\max})}(\bL^0,\bL^{-r})$ and $p^{-1}(\bL^0,\bL^{-r})^\prime(\bL^0,\bL^{-r})$ are bounded away from zero and infinity. Hence, we can rewrite the model as
\[
\by_t={\bL^{0}}\bbf_t^0+{\bL^{-r}}\times {\zero_{r_{\max}-r}}+\bepsilon_t,
\]
where $\zero_{r_{\max}-r}$ is a vector with all entries 0.
That is, the new factor loading matrix is $(\bL^0,\bL^{-r})$ while the factor scores are $(\bbf_t^{0\prime},{\bf 0}^\prime)$. Denote the $r_{\max}\times r_{\max}$ rotation matrix as $\Wb^{(0)}$, and it's still positive definite though $r_{\max}>r$. Then analogous to the proof of Lemma \ref{lemma1},  after the first iterative step, we have $\max_t\hat\bv_t^{(1)}=O_p(\log T/\sqrt{p})$, i.e.,
\begin{equation}\label{equ:fn1}
\max_t\Big\|\hat\bbf_t^{(1)}-\Wb^{(0)}(\bbf_t^{0\prime},{\bf 0}^\prime)\Big\|=O_p\bigg(\frac{\log T}{\sqrt{p}}\bigg).
\end{equation}
\vspace{1em}
\noindent\textbf{Part 2)}: In this part, given $\hat\bbf_t^{(1)}$,  we investigate the properties of $\hat\bl_i^{(1)}$.
Denote the spectral decomposition  $$T^{-1}\sum_{t=1}^T\hat\bbf_t^{(1)}{\hat\bbf_t^{(1)\prime}}=\bGamma\bLambda\bGamma^\prime,$$ where $\bLambda=\text{diag}(\lambda_1,\ldots,\lambda_{r_{\max}})$. By equation (\ref{equ:fn1}), $\lambda_j$ are of order 1 for $j\le r$ and converge to 0 for $j>r$.  Thus if we directly input the factor scores $\hat\bbf_t^{(1)}$, the $(r+1)$-th to $r_{\max}$ coordinates of the estimated loadings $\hat\bl_i^{(1)}$ will go to infinity. To overcome this issue, note that before moving to the next iteration, we always normalize the factor scores such that $\hat\bF\hat\bF^\prime/T=\Ib_{r_{\max}}$. To ease the notation and further analysis, we assume that $T^{-1}\sum\hat\bbf_t^{(1)}{\hat\bbf_t^{(1)\prime}}$ is diagonal as orthogonal rotation has no effect on $\hat\bL\hat\bF$. Suppress the superscript $(1)$ and let $\breve\bbf_t$ be the normalized factor scores, i.e.,  $\breve\bbf_t=\bLambda^{-1/2}\hat\bbf_t^{(1)} $. Further denote
\[
\bLambda_1=\text{diag}(\lambda_1,\ldots,\lambda_r),\quad\breve\bbf_t=\left(\begin{aligned}
&\breve\bbf_{1t}\\
&\breve\bbf_{2t}
\end{aligned}\right),\quad \Wb^{(0)}=\left(\begin{aligned}
&\Wb_{11}&\Wb_{12}\\
&\Wb_{21}&\Wb_{22}
\end{aligned}\right),\quad \bbf_{2t}^{0}=\Wb_{22}^{-1}\breve\bbf_{2t}.
\]
Then by (\ref{equ:fn1}), we further have that
\[
\breve\bbf_{1t}=\bLambda_1^{-1}\Wb_{11}\bbf_t^0+O_p\bigg(\frac{\log T}{\sqrt{p}}\bigg), \quad \breve\bbf_{2t}=\Wb_{22}\bbf_{2t}^0.
\]
Denote
\[
\bA:=\left(\begin{aligned}
&\bLambda_1&{\bf 0}\\
&{\bf 0}&\Ib_{r_{\max}-r}
\end{aligned}\right)^{-1}\Wb^{(0)}, \quad
\bv_t^{(1)}=\breve\bbf_t-\bA\left(\begin{aligned}
&\bbf_t^{0}\\
&\bbf_{2t}^{0}
\end{aligned}\right),\quad\bu_i^{(1)}=\bl_i-(\bA^{-1})^\prime\left(\begin{aligned}
&\bl_i^{0}\\
&{\bf 0}
\end{aligned}\right),
\]
then the diagonal entries of $\bA$ are of order 1, $\max_t\|\bv_t^{(1)}\|=O_P(\log T/\sqrt{p})$, and
\[
\bl_i^\prime\breve\bbf_t-\bl_i^{0\top}\bbf_t^{0}={\bu_i^{(1)\prime}}\breve\bbf_t+({\bl_i^{0\prime}},{\bf 0}^\prime)\bA^{-1}\bv_t^{(1)}.
\]
	Then, by the proof of Lemma \ref{lemma2}, we have
	\[
	\max_i\|\hat\bu_i^{(1)}\|=o_p(1),
	\]
	which is the desired conclusion for the first iterative step.
	
	Now we introduce the way to extend the conclusion to the iterative steps for $K\ge 2$.
	Denote the eigenvalues of $p^{-1}\sum\hat\bl_i^{(1)}{\hat\bl_i^{(1)\prime}}$ as $\beta_j$, $1\le j\le r_{\max}$. Because the diagonal entries of $\bA$ are of order 1, it's easy to verify that $\beta_j$ are of order 1 for $j\le r$ while $\beta_j=o_p(1)$ for $r<j\le r_{\max}$.  Hence,  the eigenvalue condition in Assumption 1'' (1) is not satisfied if taking $\hat\bL^{(1)}$ as the initial input. It turns out that this is not a critical problem in the following iterative steps. To illuminate this,  we can always scale the $j$-th column of $\hat\bL^{(1)}$ by a factor $\beta_j^{-1}$ for $r<j\le r_{\max}$ so that the eigenvalue conditions in Assumption 1'' (1) are satisfied. Denote the scaled loadings as $\breve\bL^{(1)}$, then
	\[
	 \breve\bL^{(1)}=\hat\bL^{(1)}\times\text{diag}(1,\cdots,1,\beta_{r+1}^{-1},\cdots,\beta_{r_{\max}}^{-1}):=\hat\bL^{(1)}\bB.
	\]
	Taking $\breve\bL^{(1)}$ as the input, then the optimized factor scores $\breve\bbf_t^{(2)}$ will converge to $(\bbf_t^{0\prime},{\bf 0}^\prime)$ by similar arguments in part 1).
	Note that the scaling of $\hat\bL^{(1)}$ will lead to a shrinkage of factor $\beta_j$ on the $j$-th column of the estimated factor score  matrix, i.e.,
	\[
	\breve\bF^{(2)}=\bB^{-1}\hat\bF^{(2)}.
	\]
	However, the normalized factor score matrices by $\breve\bF^{(2)}$ and $\hat\bF^{(2)}$ are  exactly the same. Hence, the claims in part 2) still hold and
	\[
	\max_i\|\hat\bu_i^{(K)}\|=o_p(1),\text{ for } K\ge 2,
	\]
	which concludes Theorem 2.
	Therefore, when the iteration stops, the leading $r$ eigenvalues of $\tilde \bL(r_{\max})^\prime \tilde\bL(r_{\max})/p$ are of order 1 while the others are $o_p(1)$. As a result, the eigenvalue-ratio achieves maximization asymptotically only at $k=r$ and Theorem 3 holds directly. \qed
	
}
\end{appendices}

\end{document}